\documentclass[journal,twoside,web]{IEEEtran}

\IEEEoverridecommandlockouts  

\usepackage{ulem} 

\usepackage{cite}

\usepackage{amssymb,mathtools,amsfonts,amsthm,amsmath,mathrsfs}

\usepackage{enumerate} 
\usepackage{enumitem}

\newcommand{\Rom}[1]{\MakeUppercase{\romannumeral #1}}

\usepackage{algorithm}
\usepackage[noend]{algpseudocode}

\usepackage{xcolor}
\usepackage{graphicx}
\usepackage{subfigure}
\graphicspath{{./Figure/}}

\theoremstyle{remark}
\newtheorem{remark}{Remark}
\newtheorem{lemma}{Lemma}
\newtheorem{theorem}{Theorem}
\newtheorem{proposition}{Proposition}
\newtheorem{definition}{Definition}
\newtheorem{example}{Example}

\title{
Decentralized State Estimation and Opacity Verification Based on Partially Ordered Observation Sequences}

\author{Dajiang Sun, Christoforos N. Hadjicostis, \IEEEmembership{Fellow, IEEE} and Zhiwu Li, \IEEEmembership{Fellow, IEEE}
	\thanks{Dajiang Sun is with the School of Electro-Mechanical Engineering, Xidian University, Xi'an, 710071, China (e-mail: djsun@stu.xidian.edu.cn).}%
	\thanks{Christoforos N. Hadjicostis is with the Department of Electrical and Computer Engineering, University of Cyprus, Nicosia, Cyprus (e-mail: hadjicostis.christoforos@ucy.ac.cy).}
	\thanks{Zhiwu Li is with the School of Electro-Mechanical Engineering, Xidian University, Xi'an, 710071, China, and also with the Institute of Systems Engineering, Macau University of Science and Technology, Taipa 999078, China (e-mail: zhwli@xidian.edu.cn).}
	}

\begin{document}

\pagenumbering{arabic}
\maketitle

\begin{abstract}
In this paper, we investigate state estimation and opacity verification problems within a decentralized observation architecture.
Specifically, we consider a discrete event system whose behavior is recorded by a set of observation sites. 
These sites transmit the partially ordered sequences of observations that they record to a coordinator whenever a \textit{synchronization} occurs.
To properly analyze the system behavior from the coordinator's viewpoint, we first introduce the notion of a Complete Synchronizing Sequence structure (CSS structure), which concisely captures the state evolution of each system state upon different information provided by the observation sites.
Based on the CSS structure, we then construct corresponding current-state and initial-state estimators for offline state estimation at the coordinator.
When used to verify state-isolation properties under this decentralized architecture, the use of CSS structure demonstrates a significant reduction in complexity compared with existing approaches in the literature.
In particular, we discuss how to verify initial-state opacity at the coordinator, as well as a novel opacity notion, namely current-state-at-synchronization opacity.
\end{abstract}

\section{Introduction}
\IEEEPARstart{W}{ith} the proliferation of information technology, network connectivity, and sensing capabilities, there is an increasing utilization of highly complex systems composed of multiple components that collaborate to fulfill intricate tasks, such as communication networks, autonomous vehicles, and other advanced cyber and cyber-physical systems.
These systems are often mathematically captured using discrete event systems (DESs), which are typically characterized by discrete states and event-triggered dynamics.
DESs are widely used in modeling and analyzing the higher logical behavior of these complex systems.
Due to the limited sensing capability of a plant, we may not have full information about the system’s internal state, which is needed to solve an array of control problems and/or to verify certain properties of the system. 
Therefore, the state estimation problem \cite{Ramadge1986} is of great importance in practical applications. 

Depending on the time instant at which state estimates are needed, different notions have been proposed and studied, such as current-state estimation, delayed-state estimation and initial-state estimation \cite{Hadjicostis2020, WangLafortuneLin2007, LinWangHanShe2020, HanWangLiChenChen2023}. 
Dealing with these estimation tasks depends on the observed system behavior. 
Normally, an observation architecture consists of a system itself, an observation site that observes the system, and a computational unit that analyzes the data provided by the observation site. 
However, emerging complex and large-scale systems have given rise to decentralized observation structures \cite{DeboukLafortuneTeneketzis2000}, where there may not exist a monolithic observation site that is able to record all accessible information generated by the system. 
In this paper, we concentrate on initial- and current-state estimation (and corresponding property verification problems) under a decentralized observation structure, where the information received by the coordinator is a set of partially ordered sequences of observations recorded at local sites.

State estimation serves as the basis for various so-called state isolation properties, such as diagnosability \cite{Sampath1995, YooLafortune2002,JiangHuang2001}, detectability \cite{ShuLinYing2007,YinLiWang2018}, and opacity \cite{SabooriHadjicostis2007,JacobLesageFaure2016}. 
In the absence of a monolithic observer, two main architectures---decentralized and distributed---are used to describe property verification.
Debouk \textit{et al}. \cite{DeboukLafortuneTeneketzis2000} focus on the case of two observation sites and propose three protocols for coordinated decentralized diagnosis, with the third protocol being directly linked to  \textit{co-diagnosability}. 
The algorithms for verifying co-diagnosability proposed in \cite{QiuKumar2006} have polynomial complexity with respect to (w.r.t.) the number of system states, but their overall complexity grows exponentially with the number of local diagnosers \cite{Cassez2012}.
Furthermore, the co-diagnosability of a networked DES, considering communication delays and intermittent losses of observations, is addressed in \cite{NunesMoreiraAlvesCarvalhoBasilio2018}.
The work in \cite{TakaiUshio2012} studies a decentralized fault diagnosis problem modeled by Mealy automata under state-dependence and nondeterministic output functions.
Additionally, co-diagnosability and $K$-co-diagnosability are explored within the framework of Petri nets \cite{RanSuGiuaSeatzu2018}.
The problems of decentralized diagnosis under disjunctive and conjunctive decision-making processes are analyzed in \cite{TomolaCabralCarvalhoMoreira2017, TaKaiKumar2017}.
A general distributed protocol for fault diagnosis and state estimation, called DiSIR, is studied in \cite{KeroglouHadjicostis2018} in the absence of a coordinator.
In \cite{OliveiraCabralMoreira2022}, the authors address the problem of robust decentralized diagnosis of DES against transient failures in the communication of the observations to the local diagnosers.
Li \textit{et. al.} \cite{LiHadjicostisWu2021}  deal with decentralized fault diagnosis when the information received by the coordinator involves state estimates, 
possible corrupted due to attacks on the communication channels.
For more details, readers are referred to survey papers \cite{ZaytoonLafortune2013,LafortuneLinHadjicostis2018,BasilioHadjicostisSu2021}  and book \cite{Hadjicostis2020} on this topic.

The notion of opacity introduced in \cite{Mazare2004, BryansKoutnyMazarRyan2008} for transition systems ensures that an intruder cannot determine whether a predicate representing secret information is true.
Since then, various notions related to opacity (e.g., initial-state, current-state, delayed step, and infinite-step opacity), depending on different security requirements, have drawn considerable attention in the context of DESs modeled by automata \cite{SabooriHadjicostis2007,SabooriHadjicostis2008ini,SabooriHadjicostis2011,WuLafortune2013, YinLafortune2017}. 
The verification of these notions in systems modeled by Petri nets has also been explored; see, e.g., \cite{YongLiSeatzuGiua2017, TongLanSeatzu2022,CongFantiManginiL2018, CongFantiManginiLi2019}.
The works in \cite{DongWuLi2024,YangDengQiuJiang2021} extend the notion of opacity to networked discrete event systems, considering communication losses and delays.
Recently, the concept of strong opacity is established in \cite{FalconeMarchand2015,MaYinLi2021,HanZhangZhangLiChen2023}, which focuses not only on whether a system is in a secret state at a given instant but also on whether the past visit to a secret state can be determined by an outside observer.
In \cite{PaoliLin2012}, decentralized opacity is investigated by considering whether there is a coordinator among local sites or not. In \cite{WuLafortune2013}, the authors define several notions of joint opacity based on a decentralized structure where the information sent to the coordinator from local intruders consists of  state estimates.
Zhu \textit{et. al.} in \cite{ZhuLiWu2022} introduce $K$-step opacity in centralized and decentralized structures within the framework of Petri nets.

In general, decentralized information processing depends on three components: 
1) the kind of information that a local site sends to the coordinator; 
2) the instants at which synchronizations are initiated (i.e., when the local sites send local information or decisions to the coordinator); 
3) the rules that the coordinator follows to calculate global information. The first ingredient could be a sequence of observations, a set of state estimates, or a local decision (depending on the local site’s own observation). 
The second ingredient is the synchronization strategy, whereas the last ingredient correlates with the goal that we wish to perform (e.g., fault diagnosis). 
The key feature in a decentralized information setting is that no information is sent back to or communicated among local sites, which is significantly different from a distributed observation setting where such feedback is allowed.

The problem of current-state estimation within a single synchronization step under the DO-based protocol was first studied in \cite{Hadjicostis2016Partial}, where each OS sends an observation sequence to the coordinator.
A partial-order-based method was proposed to recursively compute possible system states, assuming that the sets of observable events at all OSs are disjoint.
Our previous work \cite{SunHadjicostisLi2023} relaxed this assumption by introducing the sequence-builder, a key component for constructing the synchronizer, which estimates system states via breadth-first search.
In \cite{SunHadjicostisLi2025}, the DO-based protocol was further extended to account for bounded global and local observation errors.


The above studies focus on the online estimation problem, i.e., systematically reconstructing, within a single synchronization step, the possible totally ordered observation sequences, from which the current-state estimate can be obtained.
In contrast, this paper aims to construct state estimators from the perspective of the coordinator.
Under the DO-based protocol with a given synchronization strategy, we first define the notion of a \textit{Complete Synchronizing Sequence structure} (CSS structure), which can be used to interpret a given synchronization strategy.
The CSS structure essentially encompasses all information that may be received at the coordinator, as well as the corresponding state evolution upon receiving this information.
Consequently, once the CSS structure has been constructed, the problem of state estimation can be addressed by simply taking unions of states without a complex process at each synchronization.

Subsequently, we propose and verify state isolation properties, such as 
initial-state opacity and current-state-at-synchronization opacity under the DO-based protocol.
Current-state-at-synchronization opacity is a new notion, proposed to capture the situation where an outside observer can never be certain whether the system is in a secret state immediately after a synchronization.
The remainder of this paper is organized as follows.
In Section \Rom{2}, we review some background on language and automata theory, and formulate the DO-based protocol under a specific synchronization strategy. 
The main contributions are developed as follows.
\begin{itemize}
	\item In Section \Rom{3}, the notion of CSS structure and its properties/complexity are presented; we also construct the feasible CSS structure, from which the current-state and initial-state estimators can be obtained.
	\item In Section \Rom{4}, the verification of DO-based initial-state opacity is proposed using two approaches, one with doubly-exponential complexity and the other with exponential complexity. 
    Additionally, the concept of DO-based current-state-synchronization opacity is defined and its verification is developed.
    We also delve into the discussion on how to verify diagnosability under the DO-based protocol.
\end{itemize}

Section \Rom{5} concludes this article. 
Note that there is no overlap in contributions between this paper and our work in~\cite{SunHadjicostisLi2023} that focuses on online state estimation; this paper assumes instead a fixed synchronization strategy and focuses on the verification of properties of interest.

\section{Preliminaries}\label{sec-2}

\subsection{System Model}

Let $E$ be a finite set of symbols (events). 
A string over $E$ is a sequence of $n$ events, i.e., $s=\alpha_1\alpha_2\dots\alpha_i\dots\alpha_n$, $\alpha_i\in E$, $i\in\{1,2,\dots,n\}$.
The length of $s$ is the number of events in the sequence, denoted by $|s|$. 
We denote by $E^*$ the set of all finite-length strings over $E$, including the empty string $\epsilon$ with $|\epsilon|=0$. 
A language $L\subseteq E^*$ is a set of strings \cite{Hadjicostis2020, CassandrasLafortune2008}. 
Given strings $s$, $t$ $\in E^*$, the concatenation of $s$ and $t$ is the string $st$ (also denoted as $s\cdot t$), i.e., the sequence of symbols in $s$ followed by that in $t$. 
For any $\sigma\in E$, $s\in E^*$, we use $\sigma\in s$ to denote that $\sigma$ occurs in $s$, i.e., $s=u\sigma v$ for some strings $u,v\in E^*$.
Let $\bar{s}$ be the prefix-closure of $s$, i.e., $\bar{s}=\{t\in E^*|\exists t'\in E^*: tt'=s\}$, and $s/t$ be the symbol sequence after $t$ in $s$, i.e., for $t\in\overline{s}$, we have $t\cdot(s/t)=s$ (note that $s/s=\epsilon$ and $s/\epsilon=s$).

A DES considered in this paper is modeled as a nondeterministic finite automaton (NFA) $G=(X, E, \delta, X_0)$, where $X$ is the finite set of states, $E$ is the finite set of events, $\delta:X\times E\rightarrow 2^X$ is the next-state transition function, and $X_0$ is the set of possible initial states. 
For a set $X'\subseteq X$ and $\sigma\in E$, we define $\delta(X',\sigma)=\bigcup_{x'\in X'}\delta(x',\sigma)$;
with this notation at hand, the transition function $\delta$ can be extended recursively to $\delta^*$ (denoted by $\delta$ for the sake of brevity) whose domain is $X\times E^*$ instead of $X\times E$:
$\delta(x,\epsilon)=x$ and $\delta(x,\sigma s)=\delta(\delta(x,\sigma),s)$ for any $x\in X$, $\sigma\in E$, and $s\in E^*$. 
The system execution (or behavior) of $G$ starting from state $x$ is captured by $L(G,x)=\{s\in E^*|\delta(x, s)\neq\emptyset\}$.
For a set $X'\subseteq X$, we define $L(G,X')=\bigcup_{x'\in X'}L(G,x')$.
The system behavior generated by $G$ is described by $L(G)=L(G,X_0)$.

\subsection{Decentralized Observation-Based Information Processing}
We assume that the system is observed by $m$ OSs $O_i$ and let $\mathcal{I} = \{1, 2, \dots,m\}$ be an index set.
For $i\in \mathcal{I}$, we denote by $E_{i}$ the set of events that can be observed by $O_i$. 
The natural projection function $P_{i}:E^*\rightarrow E^*_{i}$ maps any system behavior to the sequence of observations associated with it, via
\begin{align*}
	P_{i}(\epsilon)=\epsilon 
	\quad 
	\text{and}
	\quad
	P_{i}(s\sigma)=
	\begin{cases}
		P_{i}(s)\sigma, &\text{if} \quad  \sigma\in E_{i},\\
		P_{i}(s),& \text{if} \quad \sigma\notin E_{i}.
	\end{cases}
\end{align*}
For any language $L\subseteq E^*$, define $P_i(L)=\{\omega\in E_i^*|\exists s\in L: P_i(s)=\omega\}$.
For each event $\sigma\in E$, we use $I(\sigma)=\{i\in\mathcal{I}|\sigma\in E_i\}$ to denote the index set of OSs that can observe $\sigma$.

\begin{figure}[tbp]
	\centering
	\includegraphics[scale=0.4]{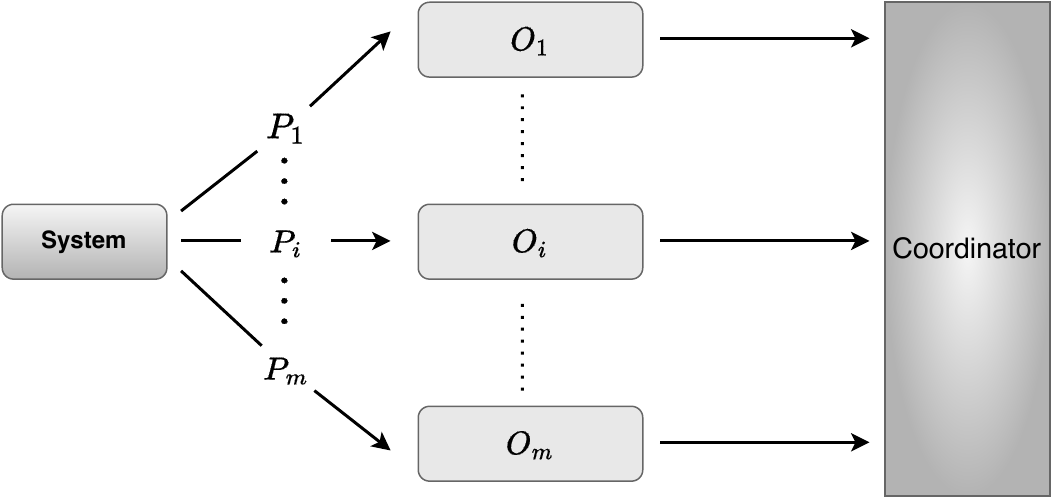}
	\caption{Decentralized observation architecture.}
	\label{architecture}
\end{figure}

The decentralized observation-based information processing, abbreviated as DO-based protocol, adopted in this work is depicted in Fig. \ref{architecture}.
We assume that OSs have no knowledge of the system model, while the coordinator knows the system model and the DO-based protocol introduced below.
Let $E_{\mathcal{I}}=\bigcup_{i\in\mathcal{I}}E_i$ denote the set of observable events from the coordinator’s perspective, and let $P_{\mathcal{I}}$ be the corresponding natural projection function.
To formally describe the DO-based protocol, we first introduce its core mechanism: \textit{synchronization}.
As previously noted, the system is monitored by OSs, each recording its own sequence of observable events.
A synchronization is a process in which the coordinator instructs all OSs to transmit their recorded sequences, and all OSs simultaneously do so.
After synchronization, each OS resets its recorded sequence to $\epsilon$ (the empty string).
As the system evolves, OSs continue recording events, and this synchronization process is repeated.
Given a system $G=(X, E, \delta, X_0)$, we define a \textit{run} as a sequence of alternating system behaviors and synchronizations: $\vec{s}=s_1{\scriptstyle\sim}s_2{\scriptstyle\sim}\dots s_j{\scriptstyle\sim}s_{j+1}$ where $s=s_1 s_2 \dots s_{j+1}\in L(G)$ and ``${\scriptstyle\sim}$'' denotes a synchronization.
We define $\tilde{s}$ as the longest prefix of $s$ after which a synchronization can be immediately initiated.
If no such prefix exists, we let $\tilde{s}=\epsilon$.
In the above run $\vec{s}$, we have $\tilde{s}=s_1 s_2 \dots s_{j}$.
The specific condition under which a synchronization is initiated will be introduced later.
We now present the DO-based protocol, which comprises three components.

1) The first component specifies the information recorded by $O_i$ and the condition under which each $O_i$ signals the coordinator to request a synchronization: Each $O_i$ is associated with a projection function $\mathcal{P}_i$ and a sync-request function $\mathcal{C}_i$,
\begin{align*}
	\mathcal{P}_i:L(G)\rightarrow E^*_i
	\quad 
	\text{and}
	\quad
	\mathcal{C}_i:\mathcal{P}_i(L(G))\rightarrow \{0,1\}.
\end{align*}
For any $s\in L(G)$, $\mathcal{P}_i(s)=P_i(s/\tilde{s})$, which records the PO-sequence observed by $O_i$ since the last synchronization (whose timing is determined in ``2) Synchronization decision'').
Function $\mathcal{C}_i$ decides whether (decision ``1'') or not (decision ``0'') $O_i$ signals the coordinator to request a synchronization.


2) The synchronization decision is the second component, which defines the condition under which the coordinator actually initiates synchronization based on the OSs’ signals. 
This decision process is governed by the function
$f:\mathcal{C}_1(\mathcal{P}_1(L(G)))\times\mathcal{C}_2(\mathcal{P}_2(L(G)))\times\dots\times\mathcal{C}_m(\mathcal{P}_m\allowbreak(L(G)))\rightarrow\{0,1\}$, 
where a decision value of ``1'' indicates that the coordinator initiates a synchronization, while a value of ``0'' implies that no such action is taken.

We refer to the functions $\{C_i\}_{i\in \mathcal{I}}$ and $f$ as the \textit{synchronization strategy}, which can be represented by a function $f_{\mathcal{I}}: L(G)\rightarrow\{0,1\}$, where, for any $s\in L(G)$, $f_{\mathcal{I}}(s)=f(\mathcal{C}_1(\mathcal{P}_1(s)),\dots,\mathcal{C}_m(\mathcal{P}_m(s)))$.
For instance, given a run $\vec{s}=s_1{\scriptstyle\sim}s_2{\scriptstyle\sim}\dots s_j{\scriptstyle\sim}s_{j+1}$, we have $f_{\mathcal{I}}(s_1 s_2 \dots s_k) = 1$ for all $k \in \{1, \dots, j\}$.
The specific realization of $\{C_i\}_{i\in \mathcal{I}}$ and $f$ considered in this paper will be formally introduced later.

3) The synchronization function: 
The coordinator receives PO-sequences provided by OSs and uses this information for state estimation or decision making. 
Specifically, when designated to perform state estimation, this task can be formalized as the synchronization function below:
\begin{align*}
	\mathcal{S}:(E_1^*\times\dots\times E_m^*)^*\times 2^X\rightarrow 2^X,
\end{align*}
which, depending on the time instant at which the set of possible states needs to be determined, results in DO-based current- or initial-state estimation.

The above three components collectively define the DO-based protocol.
As the third component is task-dependent, the DO-based protocol is denoted as $\Upsilon=(\mathcal{P},f_{\mathcal{I}})$, derived from the first two components, where $\mathcal{P}$ abstracts the set $\{\mathcal{P}_i\}_{i\in\mathcal{I}}$.

\begin{definition}\label{def-projection}
	Given a system $G$ under DO-based protocol $\Upsilon=(\mathcal{P}, f_{\mathcal{I}})$, the decentralized observation synchronization-based projection (DO-projection)  $P_{\Upsilon}:L(G)\rightarrow (E_1^*\times\dots\times E_m^*)^*$ is defined recursively as follows: for any $s\in L(G)$,
	\begin{itemize}[leftmargin=10pt]
\item $P_{\Upsilon}(s)=(\epsilon,\dots,\epsilon)$ if $\tilde{s}=\epsilon$;
\item $P_{\Upsilon}(s)=P_{\Upsilon}(\tilde{s})$ if $\tilde{s}\neq\epsilon$, where 
\begin{itemize}[leftmargin=7.8pt]
\item $P_{\Upsilon}(\tilde{s})=(P_1(\tilde{s}),\dots,P_m(\tilde{s}))$ if $\nexists s'\in\overline{\tilde{s}}\backslash\{\tilde{s}\}: f_{\mathcal{I}}(s')=1$;
\item $P_{\Upsilon}(\tilde{s})=P_{\Upsilon}(s')\allowbreak(P_1(s''),\dots,P_m(s''))$ if  $\exists s'\in\overline{\tilde{s}}\backslash\{\tilde{s}\}:f_{\mathcal{I}}(s')=1\land(\forall t\in\overline{\tilde{s}}\setminus(\overline{s'}\cup\{\tilde{s}\}):f_{\mathcal{I}}(t)=0)\land s''=\tilde{s}/ s'$.
\end{itemize}
\end{itemize}
\end{definition}

Intuitively, at each synchronization, the coordinator can only access partially ordered behaviors recorded by the $m$ OSs since the last synchronization.
This fundamental property is captured by the DO-projection $P_{\Upsilon}(s)$. 
Given a run $\vec{s}=s_1{\scriptstyle\sim} s_2{\scriptstyle\sim}\dots  s_j{\scriptstyle\sim}\allowbreak s_{j+1}$,
 $P_{\Upsilon}(s)=(P_1(s_1),\dots,P_m(s_1))\dots(P_1(s_j),\allowbreak\dots,P_m\allowbreak(s_j))$ represents the sequence of PO-sequences received by the coordinator from the start of the system execution.
We refer to the PO-sequences as the synchronization information state (SI-state), denoted by $SI(s)=(\mathcal{P}_1(s),\dots,\mathcal{P}_m(s))$.\footnote{At each synchronization, the information recorded at an OS undergoes two stages:  the sequence of observations before the synchronization and empty string $\epsilon$ after the synchronization. 
In the rest of this paper, the content of an SI-state $SI(s)$ refers to the information recorded before synchronization.}
In other words, an SI-state is the information recorded at OSs between two consecutive synchronizations.
For simplicity, we denote an SI-state as $\tau=(\tau^{(1)},\dots,\tau^{(m)})\in E_1^*\times\dots\times E_m^*$, and use $\tau_0$ to represent the initial SI-state $(\epsilon,\dots,\epsilon)$ after synchronization.
In this paper, we implement the DO-based protocol $\Upsilon_{\mathbb{N}}=(\mathcal{P},f_{\mathcal{I}})$, with the following conditions: Given a set of constants $\{\kappa_i\}_{i\in\mathcal{I}}$, $\kappa_i\in\mathbb{N}$, for any $s\in L(G)$,
\begin{enumerate}
\item $\forall i\in\mathcal{I}: \mathcal{C}_i(\mathcal{P}_i(s))=1\Leftrightarrow|\mathcal{P}_i(s)|=\kappa_i$;
\item $(\exists i\in\mathcal{I}: \mathcal{C}_i(\mathcal{P}_i(s))=1)\Leftrightarrow f_{\mathcal{I}}(s)=1$.
\end{enumerate}

Then, from  the coordinator's perspective, the set of states reachable from a state in $X'$, with $X'\subseteq X$, upon the occurrence of $\sigma\in E_{\mathcal{I}}\cup \{\epsilon\}$ is given by\footnote{We also write $\operatorname{R}_{\sigma}(x)$ when $X'=\{x\}$, i.e., when $|X'|=1$.} 	$\operatorname{R}_{\sigma}(X')=\{x\in X|\exists u\in E^*: x\in\delta(X',u)\land P_{\mathcal{I}}(u)=\sigma\}$,
whereas the unobservable reach of $X'\subseteq X$ is defined as $\operatorname{UR}(X')=\operatorname{R}_{\epsilon}(X')$.

\begin{remark}
The synchronization strategy $\Upsilon_{\mathbb{N}}$ used in this paper is chosen for its simplicity, which facilitates a clear presentation of the mechanisms and results.
However, the proposed methodologies are not restricted to this specific strategy, as the relevant definitions and methods can be readily adapted to any finite state synchronization strategy within the DO-based protocol framework.
\end{remark}

\begin{figure}[tbp]
	\centering
	\includegraphics[scale=0.76]{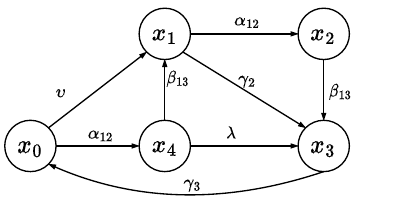}
	\caption{An NFA model $G$ under DO-based protocol, where $E_1=\{\alpha_{12}, \beta_{13}\}$, $E_2=\{\alpha_{12},\gamma_{2}\}$, and $E_3=\{\beta_{13},\gamma_{3}\}$. The subscripts of events in $E_{\mathcal{I}}=E_1\cup E_2\cup E_3$ indicate the indices of the OSs that can observe the corresponding events. Events $\upsilon$ and $\lambda$ are unobservable to all OSs.}
	\label{fig-example-1}
\end{figure}

\begin{example}\label{example-1}

Consider the system $G$ shown in Fig.~\ref{fig-example-1} where there exist three OSs, i.e., $\mathcal{I}=\{1,2,3\}$.
Let $\Upsilon_{\mathbb{N}}=(\mathcal{P}, f_{\mathcal{I}})$ be its synchronization strategy so that for all $i\in\mathcal{I}$, $\kappa_i=2$, indicating that each OS can record a sequence with at most 2 events before it signals the coordinator to initiate a synchronization.

Initially, the SI-state is $\tau_0=(\epsilon,\epsilon,\epsilon)$, indicating that no events in $E_{\mathcal{I}}$ have been executed within the system.
Now, consider the sequence $s_1=\alpha_{12}\lambda\gamma_3\alpha_{12}$, which occurs in the system, starting from state $x_0$ and ending in state $x_4$.
We analyze this process event by event:

1) When the first event of $s_1$, i.e., $\alpha_{12}$, occurs, since $\alpha_{12}\in E_1$ and $\alpha_{12}\in E_2$, both $O_1$ and $O_2$ record its occurrence. 
The resulting SI-state is $SI(\alpha_{12})=(\alpha_{12},\alpha_{12},\epsilon)$;

2) Next, when the second event of $s_1$, i.e., $\lambda$, occurs, since $\lambda\notin E_{\mathcal{I}}$, none of OSs records its occurrence; so the SI-state remains $SI(\alpha_{12}\lambda)=(\alpha_{12},\alpha_{12},\epsilon)$;

3) By analogy, when the third event of $s_1$, i.e., $\gamma_3$, occurs, the SI-state is $SI(\alpha_{12}\lambda\gamma_3)=(\alpha_{12},\alpha_{12},\gamma_3)$;

4) Finally, when the last event of $s_1$, i.e., $\alpha_{12}$, occurs, $SI(\alpha_{12}\lambda\gamma_3\alpha_{12})=(\alpha_{12}\alpha_{12},\alpha_{12}\alpha_{12},\gamma_3)$.

At this point, note that $|\alpha_{12}\alpha_{12}|=2=\kappa_1=\kappa_2$, which implies that $\mathcal{C}_1(\alpha_{12}\alpha_{12})=\mathcal{C}_2(\alpha_{12}\alpha_{12})=1$.
This causes $O_1$ and $O_2$ to signal the coordinator to initiate a \textit{synchronization}.
Consequently, function $f_{\mathcal{I}}(s_1)=1$.
As a result, the coordinator requests information from all OSs and subsequently receives SI-state $(\alpha_{12}\alpha_{12},\alpha_{12}\alpha_{12},\gamma_3)$.
Specifically, each OS sends its partially ordered sequence of observations to the coordinator: $O_1$ sends $\alpha_{12}\alpha_{12}$, $O_2$ sends $\alpha_{12}\alpha_{12}$, and $O_3$ sends $\gamma_3$.
After this transmission, the SI-state resets to $\tau_0=(\epsilon,\epsilon,\epsilon)$, indicating that no events have been observed since the last transmission of information to the coordinator.

Meanwhile, the coordinator estimates the system states based on the received SI-state.
If the coordinator knows the system starts from state $x_0$, then its set of state estimates before receiving the SI-state is $\operatorname{UR}(x_0)=\{x_0,x_1\}$.
Therefore, at this synchronization, the coordinator needs to update its estimates based on $\{x_0,x_1\}$ and SI-state $\tau=(\alpha_{12}\alpha_{12},\alpha_{12}\alpha_{12},\gamma_3)$.
This can be done by examining the possible executions in the system that match the received SI-state.
For example, the possible system behaviors starting from $x_0$ or $x_1$ could be 
$t_1=\alpha_{12}\lambda\gamma_{3}\alpha_{12}$, $t_2=\alpha_{12}\lambda\gamma_{3}\upsilon\alpha_{12}$, or $t_3=\alpha_{12}\lambda\gamma_{3}\alpha_{12}\lambda$, 
since for any $i\in\mathcal{I}=\{1,2,3\}$ $P_i(t_1)=P_i(t_2)=P_i(t_3)=\tau^{(i)}$.
Then, the state estimates obtained by the coordinator are $\delta(\{x_0,x_1\},t_1)\cup\delta(\{x_0,x_1\},t_2)\cup\delta(\{x_0,x_1\},t_3)=\{x_2,x_3,x_4\}$

Assume that the sequence $s_2=\beta_{13}\gamma_{2}\gamma_{3}\alpha_{12}$ is subsequently executed in the system. 
A synchronization happens after $\beta_{13}\gamma_{2}\gamma_{3}$ due to
$SI(s_1\beta_{13}\gamma_{2}\gamma_{3})=(\beta_{13},\gamma_{2},\beta_{13}\gamma_{3})$ 
with $|\beta_{13}\gamma_{3}|=\kappa_3$. 
In this case, the corresponding run can be represented as $\vec{s}=s_1{\scriptstyle\sim} \beta_{13}\gamma_{2}\gamma_{3} {\scriptstyle\sim}\alpha_{12}$ 
and $\tilde{s}=s_1\beta_{13}\gamma_{2}\gamma_{3}$.
Based on Definition~\ref{def-projection}, we have $P_{\Upsilon_{\mathbb{N}}}(s)=SI(s_1)SI(\beta_{13}\gamma_{2}\gamma_{3})=(\alpha_{12}\alpha_{12},\alpha_{12}\alpha_{12},\gamma_3)(\beta_{13},\gamma_{2},\beta_{13}\gamma_{3})$.~\hfill\rule{1ex}{1ex}
\end{example}

\section{State Estimation and Observer Construction Under DO-based Protocol}
In this section, we first introduce and analyze the notion of \textit{Complete Synchronizing Sequence structure} (CSS structure), which essentially encompasses all SI-states and the state evolution of each state upon any SI-state. 
Following this, we demonstrate how the CSS structure can be utilized to construct the DO-based state estimator.

In \cite{SunHadjicostisLi2023}, the notions of DO-based current-state estimation (DO-CSE) and DO-based initial-state estimation (DO-ISE) are proposed to capture the current-state estimate and initial-state estimate after the coordinator receives an SI-state $\tau$ based on the latest state estimates. 
Due to the more general scenarios analyzed in this paper,  these two notions are redefined as follows.
Given a system $G=(X,E,\delta,X_0)$ under DO-based protocol $\Upsilon_{\mathbb{N}}$, suppose that a run $\vec{s}$, i.e., a sequence of alternating system behaviors and synchronizations,  is generated by the system:
\begin{enumerate}
	\item DO-CSE aims to determine the set of states that the system possibly resides in immediately after the last synchronization, including the states that may be reached via unobservable events, i.e.,
	\begin{multline*}
		\mathcal{E}^c(P_{\Upsilon_{\mathbb{N}}}(s),X_0)=\{x\in X|\exists x_0\in X_0,\exists u\in L(G), \\ P_{\Upsilon_{\mathbb{N}}}(s)=P_{\Upsilon_{\mathbb{N}}}(u)\wedge x\in\operatorname{UR}(\delta(x_0,\tilde{u}))\},
	\end{multline*} where $\tilde{u}$ is the prefix of $u$ that is immediately followed by the last synchronization.
	
	\item DO-ISE aims to determine the set of initial states after several synchronizations, i.e.,
	\begin{multline*}
		\mathcal{E}^{i}(P_{\Upsilon_{\mathbb{N}}}(s),X_0)=\{x_0\in X_0|\exists u\in L(G,x_0):\\
		P_{\Upsilon_{\mathbb{N}}}(s)=P_{\Upsilon_{\mathbb{N}}}(u)\}. 
	\end{multline*}
\end{enumerate}
Based on the definition of synchronization, for any string $u\in L(G)$, if $\tilde{u}\neq\epsilon$, then $f_{\mathcal{I}}(\tilde{u})=1$. 
The last event in $\tilde{u}$ is always observable to at least one of the OSs. 
Thus, we use operator $\operatorname{UR}(\cdot)$ to extend the state estimates to include the unobservable reach of $\delta(x_0,\tilde{u})$ (w.r.t. the set of events that are not observable to \textit{any} OS).
This definition allows for a more comprehensive coverage of the system's possible states after a sequence of events.

\begin{remark}\label{remark-aftersync}
	Within the framework of DO-CSE, $\mathcal{E}^c(P_{\Upsilon_{\mathbb{N}}}(s),X_0)$ denotes the state estimate right after $\tilde{s}$ (note that $P_{\Upsilon_{\mathbb{N}}}(s)=P_{\Upsilon_{\mathbb{N}}}(\tilde{s})$) including the unobservable reach w.r.t. the coordinator.
	In this context, $\mathcal{E}^c(P_{\Upsilon_{\mathbb{N}}}(s),X_0)$ is equivalent to $\mathcal{E}^c(P_{\Upsilon_{\mathbb{N}}}(\tilde{s}),X_0)$, even if events in $E_{\mathcal{I}}$ occur after $\tilde{s}$.
	Consequently, the estimate $\mathcal{E}^c(P_{\Upsilon_{\mathbb{N}}}(s),X_0)$ remains valid only immediately following $\tilde{s}$. 
	The coordinator will only be able to update this estimate at the next synchronization. 
	If desirable, the coordinator could presumably attempt to predict system behaviors and potential states between synchronizations based on its knowledge of the synchronization strategy, which could be a valuable topic for future research.
\end{remark}

\subsection{Construction of CSS structure and its Properties}

The concept of \textit{S-builder} was introduced in \cite{SunHadjicostisLi2023} to provide the set of all possible system observation sequences matching a given SI-state $\tau$; state estimation is subsequently performed based on this set.
This method is applicable in online scenarios and requires no preprocessing of the synchronization strategy. 
In other words, the coordinator relies solely on $\tau$, $I(\sigma)$ for all $\sigma\in E_{\mathcal{I}}$, and the state estimates derived from the most recent synchronization.
In the sequel, we detail a preprocessing approach that can be applied to the synchronization strategy, and facilitates state estimation and observer construction, which can be important when trying to verify properties of interest, as discussed later.
This preprocessing approach can also be used for online estimation, e.g., in terms of eliminating the need for (on-line) construction of \textit{S-builders}.

We first introduce the absorbing transition function $\hat{h}$ as a mapping $\hat{h}:{E_1^*\times\dots\times E_m^*}\times E_{\mathcal{I}}\rightarrow E_1^*\times\dots\times E_m^*$, defined as follows:  for any $\tau=(\tau^{(1)},\dots,\tau^{(m)})\in E_1^*\times\dots\times E_m^*$, $\tau'=(\tau'^{(1)},\dots,\tau'^{(m)})\in E_1^*\times\dots\times E_m^*$, and $\sigma\in E_{\mathcal{I}}$, it holds 
\begin{multline*}
	\hat{h}(\tau,\sigma)=\tau'\Rightarrow \tau\models\Upsilon_{\mathbb{N}}\land (\forall i\in I(\sigma),\\
	\forall j\in\mathcal{I}/I(\sigma):(\tau'^{(i)}=\tau^{(i)}\sigma\wedge\tau'^{(j)}=\tau^{(j)})),
\end{multline*}
where the notation ``$\models$'' is explained below.
The function $\hat{h}$ is defined to capture the process by which all OSs record events during the evolution of the system between two consecutive synchronizations. 
In this context, $\tau'$ can be described as the outcome of $\tau$ when $\sigma$ occurs.
Due to the existence of the synchronization strategy $\Upsilon_{\mathbb{N}}$, the first component of the domain of $\hat{h}$ is constrained to the set of SI-states that cannot cause a synchronization, i.e., the set $\{\tau|(\tau\in E_1^*\times\dots\times E_m^*)\land(\forall i\in\mathcal{I}$: $|\tau^{(i)}|< \kappa_i)\}$, which is denoted as $\tau\models\Upsilon_{\mathbb{N}}$.
Given that synchronizations always occur as long as the events in $E_{\mathcal{I}}$ continue to be generated in the system, we refer to SI-states that are immediately followed by a synchronization and are received at the coordinator, as critical SI-states (CSI-states). For any CSI-state $\tau$, $\hat{h}(\tau,\sigma)$ is \textit{not} defined.

We now present the \textit{complete synchronizing sequence structure} (CSS structure), which serves as a fundamental component of this paper.
To clearly define CSS structure,  we first introduce a structure called \textit{synchronizing sequence structure} (SS structure) for a given state $x\in X$.

An SS structure $\mathcal{T}_x$ w.r.t. a state $x\in X$ in system $G$ under DO-based protocol $\Upsilon_{\mathbb{N}}$ is defined as $\mathcal{T}_x=(\mathcal{X}_x,T, h_{a,x},h_{r,x},E_{\mathcal{I}}\cup\{\epsilon\},T_0, \mathcal{X}_{0,x},T_c)$, where
	\begin{itemize}
		\item $\mathcal{X}_{x}\subseteq X\times L$ is the set of system states augmented with layer set $L=\{0,1,\dots,l_u\}\subset\mathbb{N}$;
		\item $T\subseteq E_1^*\times\dots\times E_m^*$ is the set of SI-states during the evolution of the system $G$;
	    \item $h_{a,x}:\mathcal{X}_{x}\times E_{\mathcal{I}}\rightarrow 2^T$ is the transition function, which is defined as follows: for any $(x',l')\in \mathcal{X}_{x}$, $\sigma\in E_{\mathcal{I}}$, and $\tau\in T$, it holds
	\begin{multline*}
		((x',l'),\sigma,\tau)\in h_{a,x}\Leftrightarrow \exists\tau'\in T,\exists \sigma'\in E_{\mathcal{I}}\cup\{\epsilon\}:\\ 
		(\tau',\sigma',(x',l'))\in h_{r,x}\land \hat{h}(\tau',\sigma)=\tau\land\operatorname{R}_{\sigma}(x')\neq\emptyset;
	\end{multline*}
	\item $h_{r,x}:T\times (E_{\mathcal{I}}\cup\{\epsilon\})\rightarrow 2^{\mathcal{X}_{x}}$ is the transition function, which is defined as follows: for any $\tau\in T$, $(x',l')\in \mathcal{X}_{x}$, and $\sigma\in E_{\mathcal{I}}\cup\{\epsilon\}$, it holds
	\begin{multline*}
		(\tau,\sigma,(x',l'))\in h_{r,x}\Leftrightarrow \mathcal{CD}_1\lor\mathcal{CD}_2,\\
		\shoveleft{\text{where}\ \mathcal{CD}_1\Leftrightarrow  \tau=\tau_0\land\sigma=\epsilon\land (x',l')=\mathcal{X}_{0,x}}\\
		\shoveleft{\text{and}\ \mathcal{CD}_2\Leftrightarrow\sigma\in E_{\mathcal{I}}\land(\exists (x'', l'')\in \mathcal{X}_{x}}:\\ 
		((x'', l''),\sigma,\tau)\in h_{a,x}\land x'\in\operatorname{R}_{\sigma}(x'')\land l'=l''+1);
	\end{multline*}
	\item $E_{\mathcal{I}}$ is the set of events recognized by the coordinator;
	\item $T_0=\tau_0\in T$ is the initial SI-state;
	\item $\mathcal{X}_{0,x}=(x,0)$ is the state $x$ with layer ``0'';
	\item $T_c\subseteq T$ is the set of CSI-states.
	\end{itemize}

We next present a simple example to illustrate the construction of SS structure, serving as a guide for the reader through the subsequent discussion and analysis and offering a comprehensive explanation of the structure.
\begin{figure}[htbp]
	\centering
	\includegraphics[scale=0.9]{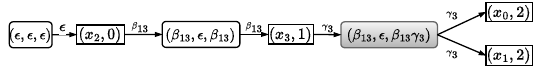}
	\caption{$\mathcal{T}_{x_2}$ of the system $G$ in Example  \ref{example-1}.}
	\label{fig-SS structure-state2}
\end{figure}
\begin{example}\label{example-t_x}
	Fig. \ref{fig-SS structure-state2} shows an SS structure w.r.t. system state ``$x_2$'' in Example  \ref{example-1}.
	The oval states, containing the sets of sequences of events, are SI-states, whereas the square states, containing the pairs of system states and numbers, form the set $\mathcal{X}_{x_2}$.
    Note that $\mathcal{X}_{0,x_2}=(x_2,0)$ and $(\beta_{13},\epsilon,\beta_{13}\gamma_3)$ is the CSI-state.\hfill\rule{1ex}{1ex}
\end{example}

An SS structure $\mathcal{T}_x$ is a multi-bipartite ordered graph that reflects the evolution of system state $x$ through various possible SI-states.
This structure is characterized by two distinct types of alternating layers—system state layers and SI-state layers—where transitions connect the states between these two layers.
System state layers contain system states, i.e., the set $\mathcal{X}_{x}$, where each state is associated with an integer that indicates its layer, distinguishing identical states across different system state layers.
SI-state layers consist of SI-states, i.e., those in set $T$, which naturally distinguish themselves across layers by sequentially absorbing possible events.
For example, in Fig. \ref{fig-SS structure-state2}, the system state layers are $\{(x_2,0)\}$, $\{(x_3,1)\}$, and $\{(x_0,2),(x_1,2)\}$, where the values ``0'', ``1'', ``2'' represent the respective layers.
The SI-states layers are $\{(\epsilon,\epsilon,\epsilon)\}$, $\{(\beta_{13},\epsilon,\beta_{13})\}$, and $\{(\beta_{13},\epsilon,\beta_{13}\gamma_3)\}$.



Functions $h_{a,x}$ and $h_{r,x}$ facilitate the transitions from system states to SI-states and from SI-states to system states, respectively.
Specifically, the SI-states, as mapped by $h_{a,x}$ from different system states augmented with the same layer, are grouped into the same SI-state layer.
Conversely, $h_{r,x}$ maps an SI-state to a subset of system states, each augmented with the same layer. 
Together, these two functions form the CSS structure structure into consecutive alternating layers, where system states and SI-states are sequentially connected. 
In detail,

1) In the definition of $h_{a,x}$: 
	$\tau'$ represents an SI-state that transitions (without requirement for a specific $\sigma'$, suggesting that any $\sigma'\in E_{\mathcal{I}}$ is acceptable) to $(x',l')$, i.e., $(\tau',\sigma',(x',l'))\in h_{r,x}$. 
	Symbol $\sigma$ is the observable event that might occur from system state $x'$.
	With these elements, a mapping from $(x',l')$ to $\tau$ is defined such that $\tau$ is the outcome of $\tau'$ when absorbing $\sigma$.
	This process signifies that the resulting SI-state, $\tau$, is determined based on the current system state being $x'$ and the occurrence of event $\sigma$, with the preceding SI-state being $\tau'$.
	For the example in Fig.~\ref{fig-SS structure-state2}, note that (i) there exists a transition relation $(\tau_0,\epsilon, (x_2,0))\in h_{r,x}$, since $\operatorname{R}_{\beta_{13}}(x_2)\neq\emptyset$ (which indicates that event $\beta_{13}$ can be the observable event occurring from $x_2$); (ii) $\hat{h}(\tau_0,\beta_{13})=(\beta_{13},\epsilon,\beta_{13})$. 
	Using (i) and (ii), we know $((x_2,0),\beta_{13},(\beta_{13},\epsilon,\beta_{13}))\in h_{a,x}$.

2) In the definition of $h_{r}$: 
    Let us first focus on Condition $\mathcal{CD}_2$ (Condition $\mathcal{CD}_1$ will be discussed later): 
	$((x'',l''),\sigma,\tau)\in h_{a,x}$ indicates that $(x'',l'')$ in the preceding system state layer (where $l'=l''+1$) transitions to SI-state $\tau$ with observable event $\sigma$.
	The set $\operatorname{R}_{\sigma}(x'')$ represents the state estimates following the occurrence of event $\sigma$ when the system state was $x''$. 
	With these elements, a mapping from $\tau$ to $(x',l')$ is defined if and only if the system could evolve from $x''$ to $x'$ upon the occurrence of observable event $\sigma$ and the SI-state $\tau$ is the outcome of absorbing $\sigma$.
	For the example in Fig.~\ref{fig-SS structure-state2}, we know that (i) there exists transition relation $((x_2,0),\beta_{13},(\beta_{13},\epsilon,\beta_{13}))\in h_{a,x}$, since $\operatorname{R}_{\beta_{13}}(x_2)=\{x_3\}$; (ii) $1=0+1$.
	Using (i) and (ii), we deduce that $((\beta_{13},\epsilon,\beta_{13}),\beta_{13}, (x_3,1))\in h_{r,x}$.

The condition $\mathcal{CD}_1$ within the definition of $h_{r,x}$ serves as the initial mapping from initial SI-state $\tau_0$ to $\mathcal{X}_{0,x}=(x,0)$, i.e., $\{\mathcal{X}_{0,x}\}=h_{r,x}(\tau_0,\epsilon)$. 
This setup is the foundation from which the entire SS structure is constructed.
Then, the SS structure is developed through the interdependence of the two functions, $h_{a,x}$ and $h_{r,x}$, where the output of one function forms the input of the other. 
Furthermore, the SI-state continuously absorbs events until it transitions into a CSI-state. 
This mechanism ensures that the size of the SS structure is finite, as it is constrained by the predicate $\hat{h}(\tau',\sigma)=\tau$ specified in the definition of $h_{a,x}$.
For example, in Fig.~\ref{fig-SS structure-state2}, we initiate a mapping as $(\tau_0,\epsilon, (x_2,0))$ in $h_{r,x}$. 
The rest of $\mathcal{T}_{x_2}$ is constructed using the function described above.
In the end, the SI-state layer terminates in the CSI-state $\{(\beta_{13},\epsilon, \beta_{13}\gamma_3)\}$ and the system state layer terminates in the set of states $\{(x,2), (x_1,2)\}$.

The above concept allows the systematic analysis of the evolution of system state $x$ upon any possible SI-state.
We are now ready to define the CSS structure.
For the sake of notational simplicity, we will continue to use the notations from the SS structure when defining the CSS structure.\footnote{Henceforth, unless otherwise specified, all related notations will refer to the CSS structure.}

\begin{definition}\label{def-CSS structure}
	A \textit{Complete Synchronizing Sequence structure} (CSS structure) $\mathcal{T}$ w.r.t.  system $G$ under DO-based protocol $\Upsilon_{\mathbb{N}}$ is an eight-tuple structure 
	\begin{align*}
		\mathcal{T}=(\mathcal{X},T, h_a,h_r,E_{\mathcal{I}}\cup\{\epsilon\},T_0, \mathcal{X}_0,T_c)
	\end{align*}
where
\begin{itemize}
	\item $T$, $E_{\mathcal{I}}$, $T_0$, and $T_c$ are the same as those in SS structure;
	\item $\mathcal{X}\subseteq X\times X\times L$ is the set of state pairs augmented with layer set $L=\{0,1,\dots,l_u\}\subset\mathbb{N}$ and $(\rho_0,\rho_1)\in X\times X$ and $\rho_d\in L$ denote the state pair component and layer component of a state $\rho=(\rho_0,\rho_1,\rho_d)\in\mathcal{X}$;
	\item $h_a:\mathcal{X}\times E_{\mathcal{I}}\rightarrow 2^T$ is the transition function, which is defined as follows: for any $\rho=(\rho_0,\rho_1, \rho_d)\in \mathcal{X}$, $\sigma\in E_{\mathcal{I}}$, and $\tau\in T$, it holds
	\begin{multline*}
		(\rho,\sigma,\tau)\in h_{a}\Leftrightarrow \exists\tau'\in T,\exists \sigma'\in E_{\mathcal{I}}\cup\{\epsilon\}:\\ 
		(\tau',\sigma',\rho)\in h_r\land \hat{h}(\tau',\sigma)=\tau\land\operatorname{R}_{\sigma}(\rho_1)\neq\emptyset;
	\end{multline*}
    \item $h_r:T\times (E_{\mathcal{I}}\cup\{\epsilon\})\rightarrow 2^{\mathcal{X}}$ is the transition function, which is defined as follows: for any $\tau\in T$, $\rho=(\rho_0,\rho_1, \rho_d)\in \mathcal{X}$, and $\sigma\in E_{\mathcal{I}}\cup\{\epsilon\}$, it holds
    \begin{multline*}
    	(\tau,\sigma,\rho)\in h_r\Leftrightarrow \mathcal{CD}_1\lor\mathcal{CD}_2,\\
    	\shoveleft{\text{where}\ \mathcal{CD}_1\Leftrightarrow \sigma=\epsilon\land \tau=\tau_0\land \rho\in\mathcal{X}_0}\\
    	\shoveleft{\text{and}\ \mathcal{CD}_2\Leftrightarrow\sigma\in E_{\mathcal{I}}\land(\exists \rho'=(\rho_0,\rho_1', \rho_d')\in \mathcal{X}:}\\ 
    	(\rho',\sigma,\tau)\in h_a\land \rho_1\in\operatorname{R}_{\sigma}(\rho_1')\land \rho_d=\rho_d'+1);
    \end{multline*}
\item $\mathcal{X}_0=\{(x,x,0)|x\in X\}$ is the set of initial state pairs with layer ``0''.
\end{itemize}
\end{definition}

Compared with the SS structure, the CSS structure comprises \textit{all} system states in a structure and exhibits the following differences.
The set $\mathcal{X}$ is now defined as the set of state pairs augmented with a layer set.
This modification is necessary to retain the corresponding initial system states at each system state layer.
Certain details in the functions $h_a$ and $h_r$ have been adjusted accordingly, although the main idea remains the same.
The primary effect of the changes to the functions is that, apart from $\mathcal{X}_0$, which is initialized first, the second component of each $\rho\in\mathcal{X}$ indicates whether it can be reached from the first component via specific system behaviors, with the corresponding SI-states connecting to the state $\rho$.
Consequently, the first element of $\rho$ in the definition of $h_r$ remains unchanged.
Overall, a CSS structure is a multi-bipartite ordered graph that reflects the evolution of all system states upon any possible SI-state. 

\begin{figure*}[htbp]
	\centering
	\includegraphics[scale=0.8]{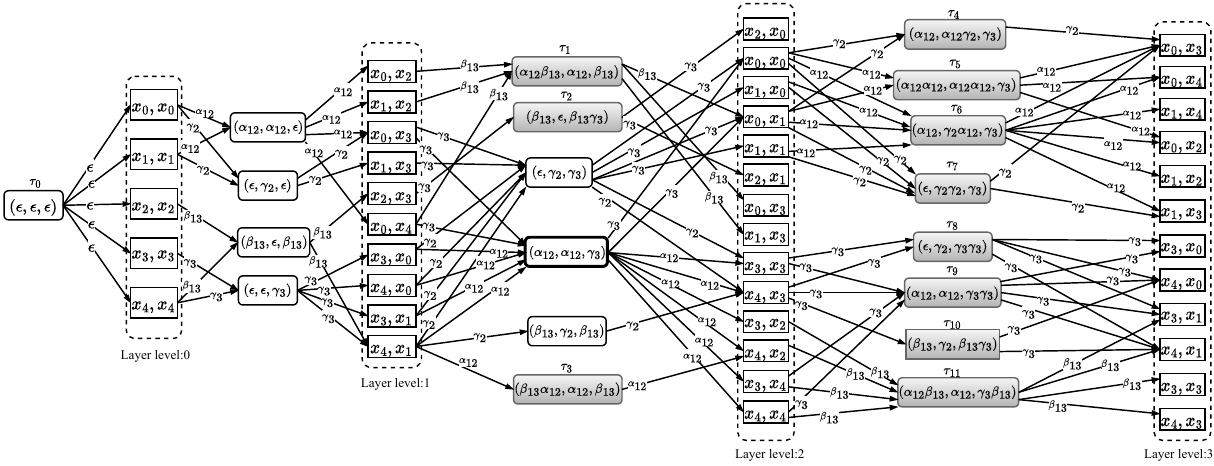}
	\caption{The CSS structure of system $G$ under DO-based protocol $\Upsilon_{\mathbb{N}}$. 
		For clarity, the sets of state pairs that share the same layer are enclosed in a dotted rectangle, with the corresponding layer listed below each rectangle. The SI-states with grey shading indicate that they are CSI-states of $T_c$, labeled from $\tau_1$ to $\tau_{11}$.}
	\label{fig-CSS-structure-1}
\end{figure*}

\begin{example}\label{example-2}
Consider again the system $G$ and the DO-based protocol $\Upsilon_{\mathbb{N}}$ shown in Fig.~\ref{fig-example-1}. 
The corresponding CSS structure is illustrated in Fig.~\ref{fig-CSS-structure-1}.
It is clear that by adding the component ``$x_0$'' to the first position of each state in $\mathcal{X}$ of $\mathcal{T}_{x_2}$ in Fig. \ref{example-t_x}, $\mathcal{T}_{x_2}$ becomes part of this structure.

From the initial SI-state $\tau_0=(\epsilon,\epsilon,\epsilon)$, based on the condition $\mathcal{CD}_1$ in Definition~\ref{def-CSS structure}, there are transition relations in $h_r$ that connect $\tau_0$ to each state pair of $\mathcal{X}_0$ via $\epsilon$, all labeled with layer 0, i.e., $(\tau_0,\epsilon, (x,x,0))\in h_r$ for all $x\in\{x_0,x_1,x_2,x_3,x_4\}$.
Similar to the way that $\mathcal{T}_{x_2}$ in Example~\ref{example-t_x} is obtained, given a state pair $\rho=(x_0,x_0,0)$, by
 $\operatorname{R}_{\alpha_{12}}(x_0)\neq\emptyset$ and $(\tau_0,\epsilon, \allowbreak\rho)\in h_r$, 
the event $\alpha_{12}$ can be absorbed by $\tau_0$ 
such that $\hat{h}(\tau_0,\allowbreak\alpha_{12})=(\alpha_{12},\alpha_{12},\epsilon)$, which leads to $(\rho,\alpha_{12},(\alpha_{12},\alpha_{12},\epsilon))\in h_a$.
As in the case of condition $\mathcal{CD}_2$ in Definition~\ref{def-CSS structure}, due to $(\rho,\alpha_{12}, (\alpha_{12},\alpha_{12},\epsilon))\in h_a$ and $\operatorname{R}_{\alpha_{12}}(x_0)\allowbreak=\{x_2,x_3,x_4\}$, the following transition relation in $h_r$ holds: $\{(x_0,x_2,1),(x_0,x_3,1),(x_0,x_4,1)\}\in h_r((\alpha_{12},\alpha_{12},\epsilon),\alpha_{12})$.
Furthermore, given that $((x_1,x_1),\alpha_{12},\allowbreak(\alpha_{12},\alpha_{12},\epsilon))\in h_a$ and $\operatorname{R}_{\alpha_{12}}(x_1)=\{x_2\}$, we also have 
$\{(x_1,x_2,1)\}\in h_r((\alpha_{12},\alpha_{12},\allowbreak\epsilon),\alpha_{12})$.
Combining these relations, we find $h_r((\alpha_{12},\alpha_{12},\epsilon),\allowbreak\alpha_{12})=\{(x_0,x_2,1),(x_0,x_3,1), (x_0,x_4,1), (x_1,x_2,1)\}$.

Let us now focus on the state pair $(x_0,x_2,1)$.
It holds that $((x_0,x_2,1), \beta_{13},\tau_1)\in h_a$ due to $((\alpha_{12},\alpha_{12},\epsilon),\alpha_{12},\allowbreak(x_0,x_2,1))\in h_r$, $\operatorname{R}_{\beta_{13}}(x_2)\neq\emptyset$, and 
$\hat{h}((\alpha_{12},\alpha_{12},\epsilon),\beta_{13})=\tau_1$.
Additionally, $(\tau_1,\beta_{13},(x_0,x_3,2))\in h_r$ because of $((x_0,x_2,1),\beta_{13},\tau_1)\in h_a$ and  $\operatorname{R}_{\beta_{13}}(x_2)=\{x_3\}$.
The rest of CSS structure can be constructed in a similar way. 
Note that, for any state pair transitioned to exclusively by CSI-states, such as $(x_1,x_3,2)$, there is no subsequent transition relation in $h_a$, based on the definition of $\hat{h}$.\hfill\rule{1ex}{1ex}
\end{example}

Since the coordinator knows the details of the synchronization strategy, it inherently knows which OSs signal the synchronization.
According to the description of the DO-based protocol, each synchronization step is triggered by the OS (or OSs) that records the last event of a system execution.
Then, a CSS structure satisfies the following property.

\begin{lemma}\label{Lemma-CSS structure-lastEvent}
Let $\mathcal{T}=(\mathcal{X},T, h_a,h_r,E_{\mathcal{I}}\cup\{\epsilon\},T_0, \mathcal{X}_0,T_c)$ be a CSS structure. 
Then, it holds
		\begin{align*}
			\forall \tau\in T_c, \forall (\tau,\sigma,\rho),(\tau,\sigma',\rho')\in h_r: \sigma=\sigma'.
		\end{align*}
\end{lemma}
\begin{proof}
	We prove this by contradiction.
	Assume there exist $\tau\in T_c$ and $(\tau,\sigma,\rho),(\tau,\sigma',\rho')\in h_r$, such that $\sigma\neq\sigma'$.
	Then, there exist two SI-states $\tau_1,\tau_2\in T$, such that $\hat{h}(\tau_1,\sigma)=\tau$ and $\hat{h}(\tau_2,\sigma')=\tau$.
	If $I(\sigma)=I(\sigma')$, then for any $i\in I(\sigma)$, it holds $\tau_1^{(i)}\sigma=\tau^{(i)}=\tau_2^{(i)}\sigma'$ which contradicts the assumption that $\sigma\neq\sigma'$.
	If $I(\sigma)\neq I(\sigma')$, then either $I(\sigma)\backslash I(\sigma')\neq \emptyset$ or $I(\sigma')\backslash I(\sigma)\neq \emptyset$ holds.
	Let us assume $I(\sigma)\backslash I(\sigma')\neq \emptyset$.
	For any $j\in I(\sigma)\setminus I(\sigma')$, we have $\tau^{(j)}_1\sigma=\tau^{(j)}=\tau^{(j)}_2$, which indicates $|\tau^{(j)}_1\sigma|=|\tau^{(j)}_2|$.
    Since the synchronization is triggered by the OS that records the last event, i.e., event $\sigma$ or $\sigma'$, then $|\tau^{(j)}_1\sigma|=\kappa_j=|\tau^{(j)}_2|$, which contradicts the fact that $\tau_2\models\Upsilon_{\mathbb{N}}$, concluding that $\sigma=\sigma'$.
\end{proof}

Lemma~\ref{Lemma-CSS structure-lastEvent} indicates that the transitions originating from a CSI-state in a CSS structure share the same event label.
Let $M(\tau)$ denote the set of state pairs which are transitioned to by $\tau$ in a CSS structure, i.e., for any event $\sigma\in E_{\mathcal{I}}$ and for any state $\rho$ in $\mathcal{X}$, we have $M(\tau)=\{(\rho_0,\rho_1)|(\tau,\sigma,\rho)\in h_r\}$.

As seen earlier, the CSS structure is used to interpret the given DO-based protocol, allowing us to explore the state evolution of each state upon any SI-state. Therefore, given a system, the state evolution is decoded in its corresponding CSS structure, as concluded below.

\begin{proposition}\label{pro-current-string}
	Let $\mathcal{T}=(\mathcal{X},T, h_a,h_r,E_{\mathcal{I}}\cup\{\epsilon\},T_0, \mathcal{X}_0,T_c)$ be the CSS structure w.r.t. system $G$ under DO-based protocol $\Upsilon_{\mathbb{N}}$. The following hold:
	\begin{enumerate}
		\item $\forall \tau\in T\setminus\{\tau_0\},\forall x,x'\in X:((x,x')\in M(\tau)\Leftrightarrow (\exists t\in L(G,x):x'\in\delta(x,t)\land (\forall i\in\mathcal{I}:P_i(t)=\tau^{(i)})))$;
		\item $\forall s\in L(G):SI(s)\in T \land SI(\tilde{s})\in T_c\cup\{\tau_0\}$. 
	\end{enumerate}
\end{proposition}
\begin{proof}
	1) ($\Rightarrow$)
	If $(x,x')\in M(\tau)$, according to the definition of CSS structure, there exist a sequence of events $\omega=\sigma_1\sigma_2\dots\sigma_n$, a sequence of SI-states $\tau_0\tau_1\dots\tau_n$, and a sequence of states $x_0x_1x_2\dots x_n$, such that $\tau_0=(\epsilon,\dots,\epsilon)$, $\tau_n=\tau$, $x_0=x$, $x_n=x'$ and for any $j\in\{1,\dots,n\}$, we have $((x,x_{j-1},j-1),\sigma_j,\tau_j)\in h_a)$ and $(\tau_j,\sigma_j,(x,x_j,j))\in h_r$ where $\tau_j=\hat{h}(\tau_{j-1},\sigma_j)$ and $x_j\in\operatorname{R}_{\sigma_j}(x_{j-1})$.
	Since $\tau\in T\setminus\{\tau_0\}$, it follows that $n>0$ in this context.
	Then, there exists $t\in L(G,x)$ such that for all $i\in\mathcal{I}$, $P_i(t)=\tau^{(i)}$. 
	In accordance with the definition of the operator $\operatorname{R}(\cdot)$, it follows that $x'\in\delta(x,t)$ which appropriately signifies the ``$\Rightarrow$'' direction.
	
	The direction of ``$\Leftarrow$'' can be proven similarly based on the definition of CSS structure, which completes the proof of 1) in this proposition.
	
	2) To prove the truth of the predicate ``$\forall s\in L(G):SI(s)\in T$'', we proceed as follows.
	Given a string $s\in L(G)$, we know $SI(s)=(\mathcal{P}_1(s),\dots,\mathcal{P}_m(s))=(P_1(s\backslash\tilde{s}),\dots,P_m(s\backslash\tilde{s}))$.
	Then, there exist $x,x'\in X$ such that $x\in \delta(X_0,\tilde{s})$ and $x'\in\delta(x,s\backslash\tilde{s})$.
	By $(x,x,0)\in\mathcal{X}_0$, 
	similarly to the proof of 1) in this proposition, there exists an SI-state $\tau\in T$ such that for any $i\in\mathcal{I}$, $P_i(s\backslash\tilde{s})=\tau^{(i)}$.
    
    If there is no synchronization within the sequence $s$, it follows that $SI(\tilde{s})=\tau_0$;
    otherwise, $\tilde{s}\neq\epsilon$ and $SI(\tilde{s})\in T$.
    Since the string $\tilde{s}$ is immediately followed by a synchronization, $SI(\tilde{s})$ is a CSI-state, i.e., there does not exist an event  $\sigma\in E_{\mathcal{I}}$ such that $\hat{h}(SI(\tilde{s}),\sigma)$ is defined; we conclude that $SI(\tilde{s})\in T_c$.
\end{proof}

Proposition \ref{pro-current-string}.2) indicates that the set $T$ in the system's CSS structure $\mathcal{T}$ contains all possible SI-states, including CSI-states, during the system evolution.
Combined with the conclusion in Proposition \ref{pro-current-string}.1), the CSS structure can be used to perform DO-based current-state estimation at one synchronization step, based on the assumption that the system is known to be in a subset of states $X'$.
This process is described in the following result.

\begin{proposition}\label{pro-current-entend}
Let $\mathcal{T}=(\mathcal{X},T, h_a,h_r,E_{\mathcal{I}}\cup\{\epsilon\},T_0, \mathcal{X}_0,T_c)$ be the CSS structure of a system $G$ under DO-based protocol $\Upsilon_{\mathbb{N}}$. 
Given $s\in L(G,X')$ such that $P_{\Upsilon_{\mathbb{N}}}(s)=\tau$ where $\tau\in T_c$, i.e., only one synchronization happens, it holds
\begin{align*}
\mathcal{E}^c(\tau,X')=\{x\in X|\exists x'\in X', (x',x)\in M(\tau)\}
\end{align*}
\end{proposition}
\begin{proof}
The proof follows directly from Lemma \ref{Lemma-CSS structure-lastEvent} and Proposition \ref{pro-current-string}.
\end{proof}

The following proposition presents the method to perform DO-based initial-state estimation at one synchronization step.

\begin{proposition}\label{pro-ini}
Let $\mathcal{T}=(\mathcal{X},T, h_a,h_r,E_{\mathcal{I}}\cup\{\epsilon\},T_0, \mathcal{X}_0,T_c)$ be a CSS structure of a system $G$ under DO-based protocol $\Upsilon_{\mathbb{N}}$. 
Given $s\in L(G,X_0)$ such that $P_{\Upsilon_{\mathbb{N}}}(s)=\tau$ where $\tau\in T_c$, i.e., only one synchronization happens, it holds
\begin{align*}
	     \mathcal{E}^i(\tau,X_0)=\{x\in X_0|\exists x'\in X, (x,x')\in M(\tau)\}.
\end{align*}

\end{proposition}
\begin{proof}
According to the definition of $h_r$ in Definition \ref{def-CSS structure}, the first element of $\rho$ remains unchanged.
With this understanding, this proof follows directly from Proposition \ref{pro-current-string}.
\end{proof}

\begin{example}
	Let us consider the system in Example \ref{example-2}. 
	When the sequence $s_1=\alpha_{12}\lambda\gamma_3\alpha_{12}$ is generated from state $x_0$, the SI-states are as follows, in chronological order: $(\alpha_{12},\alpha_{12},\epsilon)$, $(\alpha_{12},\alpha_{12},\gamma_3)$, $(\alpha_{12}\alpha_{12},\alpha_{12}\alpha_{12},\gamma_3)$. 
	Synchronization occurs after the sequence $s_1$, and the coordinator receives $\tau_5=(\alpha_{12}\alpha_{12},\alpha_{12}\alpha_{12},\gamma_3)$. 
	Based on the CSS structure shown in Fig.~\ref{fig-CSS-structure-1}: 1) If we aim to estimate the current-state of the system, the corresponding state estimates are $\{x_2,x_3,x_4\}$, since $(x_0,x_2,3)$, $(x_0,x_3,3)$, and $(x_0,x_4,3)$ are transitioned to by $\tau_5$.
	2) If we aim to estimate the initial-state of the system solely based on the sequence $s_1$, the corresponding initial-state estimate is $\{x_0\}$, since the first state component of the state pairs that are transitioned to by $\tau_5$ is $x_0$.\hfill\rule{1ex}{1ex}
	
\end{example}

As the CSS structure $\mathcal{T}$ is constructed from any state $x\in X$, we can obtain the state estimates from any system state upon any possible CSI-state.
However, this may increase the complexity of building such a structure, since it is not necessary to identify all possible CSI-states from every system state.
In other words, some CSI-states in $\mathcal{T}$ may not be received by the coordinator during the system evolution.
In Algorithm~\ref{algo1-CSS structure}, we demonstrate how to construct this structure such that only the CSI-states that can actually be received by the coordinator during the system evolution are computed. 
We denote this structure as $\mathcal{T}_{fe}$ (feasible CSS structure).

In Algorithm~\ref{algo1-CSS structure}, the set $\mathcal{X}_0$ is initialized with pairs of identical system initial states.
Starting from $\mathcal{X}_0$, the structure is expanded via a breadth-first search that alternates between SI-state layers and state pair layers.
The variable $l_u$ denotes the current depth of state pair layers in this search.
The procedures $\Call{SI-State}{\tau}$ and $\Call{State-SI}{\rho}$ explore successor state pairs and SI-states, respectively, building the transition relations in $h_r$ and $h_a$.
In lines \ref{algo1-add to ini1}--\ref{algo1-add to ini2}, 
state $x$ represents a possible state estimate under CSI-state $\tau$, and $(x, x, 0)$ is added to $\mathcal{X}_0$ to allow further CSI-state exploration in the next synchronization.
This mechanism enables efficient construction of the feasible CSS structure, particularly when $\mathcal{X}_0$ is initialized as $\{(x_0,x_0,0)|x_0\in X_0\}$ to avoid unnecessary computation.

\begin{algorithm}[t]
	\caption{Construction of Feasible CSS Structure.}\label{algo1-CSS structure}
	\begin{algorithmic}[1]
		\Require System $G$ and DO-based protocol $\Upsilon_{\mathbb{N}}$.
		\Ensure  A feasible CSS structure $\mathcal{T}_{fe}=(\mathcal{X},T,h_a,h_r,E_{\mathcal{I}}\cup\{\epsilon\},T_0, \allowbreak \mathcal{X}_0,T_c)$.
		\State $\mathcal{X}_0\gets \{(x_0,x_0,0)|x_0\in X_0\}$, $T\gets T_0$, $\mathcal{X}\gets\mathcal{X}_0$, 
		$T_0\gets \{\tau_0\}$, $h_a\gets\emptyset$, $h_r\gets\emptyset$, $T_{in}\gets\{\tau_0\}$;
		\While{$\exists\rho\in\mathcal{X}_0$ that has not been examined}
		\For{$\rho\in\mathcal{X}_0$ that has not been examined} \label{algo1-examine-ini}
		\State $l_u\gets 0$ and $h_r\gets h_r\cup\{(\tau_0, \epsilon,\rho)\}$;
		\EndFor
		\While{$T_{in}\neq\emptyset$}
		\State $T_{in}\gets\emptyset$;
		\For{$\rho\in\mathcal{X}$, s.t. $\rho_d=l_u$, that has not been examined}
		\State $\Call{State-SI}{\rho}$, $T_{in}\gets T_{in}\cup\{\tau|(\rho,\sigma,\tau)\in h_a\}$;
		\EndFor
		\For{$\tau\in T_{in}$}
		\State  $\Call{SI-State}{\tau}$;
		\EndFor
		\State $l_u\gets l_u+1$;
		\EndWhile
		\EndWhile
		\State Return $\mathcal{T}_{fe}$; 
		\Procedure{State-SI}{$\rho$}
		\For{$(\tau',\sigma',\rho)\in h_r$ that has not been examined}
		\State $h_a\gets \{(\rho,\sigma,\tau)|\exists\sigma\in E_{\mathcal{I}}$: $\hat{h}(\tau',\sigma)=\tau\land\operatorname{R}_{\sigma}(\rho_1)\neq\emptyset\}$, $T\gets T\cup\{\tau\}$;
		\If{$\tau$ is a CSI-state}
		\State $T_c\gets T_c\cup\{\tau\}$;
		\EndIf
		\EndFor
		\EndProcedure
		\Procedure{SI-State}{$\tau$}
		\For{$(\rho,\sigma,\tau)\in h_a$ that has not been examined}
		\State $h_r\gets h_r\cup\{(\tau, \sigma, (\rho_0,x,\rho_d+1))|x\in\operatorname{R}_{\sigma}(\rho_1)\}$, $\mathcal{X}\gets \mathcal{X}\cup\{(\rho_0,x,\rho_d+1)|x\in\operatorname{R}_{\sigma}(\rho_1)\}$, $L\gets L\cup\{\rho_d+1\}$;\label{algo1-add to ini1}
		\If{$\tau\in T_c$}
		\State $\mathcal{X}_0\gets\mathcal{X}_0\cup\{(x,x,0)|x\in\operatorname{R}_{\sigma}(\rho_1)\}$; \label{algo1-add to ini2}
		\EndIf
		\EndFor
		\EndProcedure
	\end{algorithmic}
\end{algorithm}

\begin{proposition}
	Let $\mathcal{T}_{fe}=(\mathcal{X},T, h_a,h_r,E_{\mathcal{I}}\cup\{\epsilon\},T_0, \mathcal{X}_0,T_c)$ be the feasible CSS structure of a system $G$ under DO-based protocol $\Upsilon_{\mathbb{N}}$. The following hold:
	\begin{enumerate}
		\item $\bigcup_{s\in L(G)}SI(s)=T$; 
		\item $\bigcup_{s\in L(G)}SI(\tilde{s})=\{\tau_0\}\cup T_c$.
	\end{enumerate}
\end{proposition}
\begin{proof}
This follows directly from Algorithm~\ref{algo1-CSS structure}, by induction on the length of $P_{\mathcal{I}}(s)$.
\end{proof}

\begin{remark}
Different synchronization strategies lead to different CSS structure structures, characterized by the function $\hat{h}$ in the definition of $h_a$. 
At first glance, the computation complexity of both DO-based current-state estimation and DO-based initial-state estimation (simply taking the union of states in precomputed feasible CSS structure) at one synchronization step appears lower than the online approach described in \cite{SunHadjicostisLi2023}. 
However, this does not negate the necessity of the online method. 
Note that the CSS structure must be constructed beforehand when employing state estimation using the method in this paper. 
This introduces additional computational costs compared with the one-time online method.  
The online method offers more flexibility in handling situations, where the synchronization strategy is not fixed or not predetermined. 
For example, in one synchronization step,  $|\mathcal{P}_i(s)|\leq 2$ for all $i\in\mathcal{I}$ in $1)$ in Section \ref{sec-2}.B, while in the next synchronization step, it could change to $|\mathcal{P}_i(s)|\leq 8$ for all $i\in\mathcal{I}$ to minimize the costs associated with communication.
However, as analyzed later, when a DO-based protocol employs a fixed synchronization strategy, the CSS structure can be effectively utilized to construct state-estimators (observers).
\end{remark}

\subsection{Current-State and Initial-State Estimator Construction}

Since a CSS structure reflects the system evolution from any system state to possible system states through potential CSI-states, it can be employed to recursively solve both current and initial state estimation problems upon a sequence of CSI-states in accordance with Propositions \ref{pro-current-entend} and \ref{pro-ini}, respectively.
Furthermore, in many applications, such as diagnosability and opacity, it is essential to have detailed knowledge of all state estimates prior to the system's operation, particularly for verification purposes.
In light of this, we introduce the notions of DO-based current-state estimator (DO-observer) and DO-based initial-state estimator below.

\begin{definition}\label{DO-observer}
Consider system $G=(X,E, \delta, X_0)$ under DO-based protocol $\Upsilon_{\mathbb{N}}$, and let $\mathcal{T}_{fe}=(\mathcal{X},T,h_a,h_r,E_{\mathcal{I}}\cup\{\epsilon\},\allowbreak T_0, \mathcal{X}_0,T_c)$ be its feasible CSS structure. 
The DO-observer of $G$ is a deterministic finite automaton (DFA) $G^{\Upsilon_{\mathbb{N}}}_{obs}=Ac(X_{obs},T_c,f_{obs},\allowbreak X_{obs,0})$\footnote{$Ac(\cdot)$ is the accessible part of $\cdot$, i.e., the part that can be reached from the initial state.} where
	\begin{itemize}
		\item $X_{obs}\subseteq 2^X\setminus\{\emptyset\}$ is the set of states;
		\item $T_c$ is the set of CSI-states, serving as the transition events;
		\item $X_{0,obs}=\operatorname{UR}(X_0)\in X_{obs}$ is the initial state;
		\item $f_{obs}:X_{obs}\times T_c\rightarrow X_{obs}$ is the transition function defined as: for any $q\in X_{obs}$ and $\tau\in T_c$, it holds $f_{obs}(q,\tau)=\{x|\exists x'\in q: x\in \mathcal{E}^c(\tau,x')\}$.
	\end{itemize}
\end{definition}

Similar to a standard observer \cite{Hadjicostis2020}, the DO-observer contains all possible state estimates following any sequence of CSI-states starting from the initial state set $X_{0,obs}$. 
The key difference, however, lies in the driving events: in the DO-observer, the transitions from one state to another are driven by elements in $T_c$, rather than (observable) events in $E_{\mathcal{I}}$. 
This arrangement is logical, as the coordinator can only update its state estimate upon receiving synchronization information, i.e., a CSI-state.

Note that the initial state of the DO-observer is assigned as the unobservable reach of the system initial states. 
Actually, every subsequent state is assigned as the unobservable reach of the set of state estimates immediately following the last synchronization. 
These settings are implemented since the coordinator updates its state estimate based on the state estimates prior to the occurrence of a CSI-state.

Function $f_{obs}$ operates as follows: based on the current state $q$, it explores any $\tau$ in $T_c$, considered as its synchronization information, such that $\mathcal{E}^c(\tau,q)$ is not empty. 
The result is then adopted as the next state of the estimator.
The domain of function $f_{obs}$ can be extended to $X_{obs}\times T_c^*$ in the standard recursive manner: $f_{obs}(q,\lambda\tau)=f_{obs}(f_{obs}(q,\lambda),\tau)$ for $q\in X_{obs}$, $\lambda\in T_c^*$, and $\tau\in T_c$ ($f_{obs}(q,\lambda\tau)$ is undefined if $f_{obs}(q,\lambda)$ is undefined).
\begin{proposition}\label{pro-obs}
	Given a system $G$ and its DO-observer $G^{\Upsilon_{\mathbb{N}}}_{obs}$, we have (i) $\bigcup_{s\in L(G)}{P_{\Upsilon_{\mathbb{N}}}(s)}= L(G^{\Upsilon_{\mathbb{N}}}_{obs})$, and (ii) for any $s\in L(G)$, $\mathcal{E}^c(P_{\Upsilon_{\mathbb{N}}}(s),X_0)=f_{obs}(X_{0,obs},P_{\Upsilon_{\mathbb{N}}}(s))$.
\end{proposition}
\begin{proof}
We first prove (i).
	
($\subseteq$) 
Given a string $s\in L(G)$, its corresponding run is $\vec{s}=s_1{\scriptstyle\sim}s_2{\scriptstyle\sim} \dots {\scriptstyle\sim}s_j{\scriptstyle\sim} \dots {\scriptstyle\sim} s_n{\scriptstyle\sim}s_{n+1}$, where $s_j=s_{j1}s_{j2}\dots s_{jk_{j}}$, for any $j\in\{1,2,\dots,n\}$ and $|s_j|=k_j$.
Then, there exist $x_0\in X_0$, $x_j\in X$,  such that $x_j\in\delta(x_{j-1},s_j)$.
According to the definition of function $f_{obs}$, there exist $\iota\in T^*_c$, and $X_j\in 2^X$, for any $j\in\{1,2,\dots,n\}$, such that $\iota=(P_1(s_1),\dots,P_m(s_1))\dots(P_1(s_j),\dots,P_m(s_j))\dots(P_1(s_n),\\\dots, P_m(s_n))$, $f_{obs}(X_{j-1},(P_1(s_j), \dots,P_m(s_j)))=X_j$, and $x_j\in X_j$.
Together, these indicate $\iota\in L(G^{\Upsilon_{\mathbb{N}}}_{obs})$, thereby completing the proof of the ($\subseteq$)  relation.
	
The reverse ($\supseteq$) can be proven similarly, which completes the proof of (i).
	
We prove (ii) by induction on the length of $|P_{\Upsilon_{\mathbb{N}}}(s)|$, i.e.,  the number of CSI-states received by the coordinator.
Clearly, the basic induction step holds: when $|P_{\Upsilon_{\mathbb{N}}}(s)|=0$ (i.e., $\tilde{s}=\epsilon$), according to the definition of DO-CES, $\mathcal{E}^c(\epsilon,X_0)=\operatorname{UR}(X_0)=X_{0,obs}$, i.e., the initial state of $G^{\Upsilon_{\mathbb{N}}}_{obs}$.
	
For the induction step, we assume that the induction hypothesis is true when $|P_{\Upsilon_{\mathbb{N}}}(s)|=k$, $k\in\mathbb{N}$.
Then, let us consider the case when $|P_{\Upsilon_{\mathbb{N}}}(s)|=k+1$.
Suppose that $s=s_1 {\scriptstyle\sim} s_2{\scriptstyle\sim}  \dots \allowbreak{\scriptstyle\sim} s_j {\scriptstyle\sim} \dots {\scriptstyle\sim} s_k{\scriptstyle\sim}s_{k+1}{\scriptstyle\sim}s_{k+2}$ and $t=s_1 {\scriptstyle\sim} s_2 \dots {\scriptstyle\sim} s_j {\scriptstyle\sim} \dots {\scriptstyle\sim} s_k$.	
According to the inductive assumption,  
$\mathcal{E}^c(P_{\Upsilon_{\mathbb{N}}}(t),X_0)=f_{obs}\allowbreak(X_{0,obs},P_{\Upsilon_{\mathbb{N}}}(t))$.
Let $\tau=(P_1(s_{k+1}),\dots,P_m(s_{k+1}))$.
Based on Proposition~\ref{pro-current-entend}  and the function $f_{obs}$, we know
$\mathcal{E}^c(\tau,\mathcal{E}^c(P_{\Upsilon_{\mathbb{N}}}(t), X_0))=\allowbreak f_{obs}(f_{obs}(X_{0,obs},P_{\Upsilon_{\mathbb{N}}}(t)),\tau)$, thereby completing the proof.
\end{proof}

%
%

\begin{figure}[htbp]
	\centering
	\includegraphics[scale=0.6]{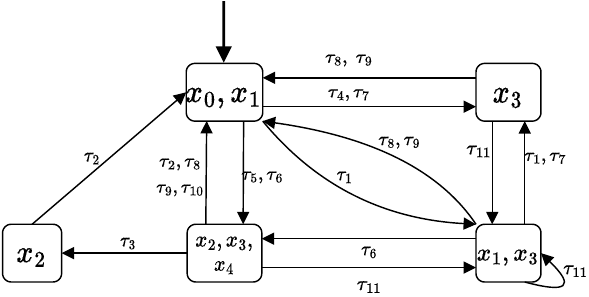}
	\caption{DO-observer for $G$, where the transitions labeled as ``$\tau$'' correspond to the labels shown in Fig.~\ref{fig-CSS-structure-1}.}
	\label{DO-observer-fig}
\end{figure}
\begin{example}
Let us revisit system $G$ in Example \ref{example-2}. 
Since its feasible CSS structure is identical to that shown in Fig.~\ref{fig-CSS-structure-1}, we refer to it as the feasible CSS structure hereafter. 
Assume that the initial state set is $\{x_0,x_1\}$. 
The DO-observer for $G$ is depicted in Fig.~\ref{DO-observer-fig}, with the initial state being $\{x_0,x_1\}$.
Suppose that at the first synchronization, the coordinator receives the CSI-state $\tau_5=(\alpha_{12}\alpha_{12},\alpha_{12}\alpha_{12},\gamma_3)$.
From Fig.~\ref{fig-CSS-structure-1}, we obtain $M(\tau_5)=\{(x_0,x_3),(x_0,x_4),(x_0,x_2)\}$.
According to Proposition~\ref{pro-current-entend}, the current-state estimate of the system and the corresponding DO-observer's state is given by $f(\{x_0,x_1\}, \tau_5) = \mathcal{E}^c(\tau_5, \{x_0,x_1\}) = \{x_2, x_3, x_4\}$.
Similarly, if at the second synchronization, the coordinator receives the CSI-state $\tau_3$, the current-state estimate and the corresponding state in the DO-observer becomes $\{x_3\}$.
In summary, from the system's initial state, by receiving the sequence of CSI-states $\tau_5\tau_3$, the coordinator infers that the current-state estimate is
$f_{obs}(\{x_0,x_1\}, \tau_5\tau_3) = {x_3}$.
The remaining part of the DO-observer can be constructed accordingly.
\hfill\rule{1ex}{1ex}
\end{example}
\begin{remark}
Given that our focus is on interpreting the synchronization strategy, specifically the construction of a CSS structure and the verification of relevant properties from the coordinator’s perspective, the DO-based current-state estimation is defined to concentrate on the instant when synchronization occurs, even if multiple observations are generated between synchronizations.
This notion, along with the DO-observer, can be clearly extended to account for what happens between synchronizations based on different purposes, which correlates with the discussion in Remark \ref{remark-aftersync}.

\end{remark}
Next, we define the notion of a DO-based initial-state estimator, which aligns with the trellis-based approach discussed in the literature \cite{SabooriHadjicostis2008ini,Hadjicostis2020}.
\begin{definition}
	Consider system $G=(X,E, \delta, X_0)$ under DO-based protocol $\Upsilon_{\mathbb{N}}$, and let $\mathcal{T}_{fe}=(\mathcal{X},T, h_a,h_r,E_{\mathcal{I}}\cup\{\epsilon\},T_0, \mathcal{X}_0,T_c)$ be its feasible CSS structure. 
	The DO-based initial-state estimator of $G$ is a DFA $G^{\Upsilon_{\mathbb{N}}}_{i,obs}=Ac(X^{i}_{obs},T_c,f^i_{obs},X^{i}_{0,obs})$ where 
	\begin{itemize}
		\item $X^i_{obs}\subseteq 2^{X\times X}\backslash\{\emptyset\}$ is the set of states;
		\item $T_c$ is the set of CSI-states, serving as the transition events;
		\item $X^{i}_{0,obs}=\{(x_0,x_0)|x_0\in X_0\}$ is the initial state;
		\item $f^i_{obs}:X^i_{obs}\times T_c\rightarrow X^i_{obs}$ is the transition function defined as: for any $m\in X^i_{obs}$ and $\tau\in T_c$, it holds $f^i_{obs}(m,\tau)=\{(x_1,x_3)|\exists (x_1,x_2)\in m,(x_2,x_3)\in M(\tau)\}.$
	\end{itemize}
\end{definition}
The initial state $X^{i}_{0,obs}$ is the set of pairs of identical initial system states.
Each element in any other state $m\in X^i_{obs}\backslash \{X^{i}_{0,obs}\}$ contains the initial state component and the corresponding reachable current state component upon a sequence of CSI-states so far.
This mechanism is implemented by function $f^i_{obs}$.
Like function $f_{obs}$, the domain of $f^i_{obs}$ can be extended to $X^i_{obs}\times T_c^*$ in a similar manner. For any state $m\in X^i_{obs}$, we use $F(m)$ to denote the set of first components in $m$, i.e., $F(m)=\{x|\exists x'\in X:(x,x')\in m\}$. 
Then, we have the following result.
\begin{proposition}\label{prop-ini}
Given a system $G$ and its  DO-based initial-state estimator $G^{\Upsilon_{\mathbb{N}}}_{i,obs}$,  
$\bigcup_{s\in L(G)}{P_{\Upsilon_{\mathbb{N}}}(s)}=L(G^{\Upsilon_{\mathbb{N}}}_{i,obs})$ and $		\mathcal{E}^{i}(P_{\Upsilon_{\mathbb{N}}}(s),X_0)=F(f^i_{obs}(X^i_{0,obs},P_{\Upsilon_{\mathbb{N}}}(s)))$ hold.
\end{proposition}
\begin{proof}
The proof of this proposition employs the same strategy as that utilized in the proof of Proposition \ref{pro-obs}, and is therefore omitted here.
\end{proof}
\subsection{Complexity Analysis}
Given a system $G=(X,E,\delta,X_0)$ under DO-based protocol $\Upsilon_{\mathbb{N}}$, in the worst case, the maximum number of sequences of observations observed by $O_i$ is $\Delta_i:=|E_i|^0+|E_i|+\dots+|E_i|^j+\dots+|E_i|^{\kappa_i}$, where $j\in\{0,1,\dots,\kappa_i\}$ and $|E_i|^j$ is number of sequences of length $j$.
The number of sequences observed by $O_i$ that cannot (can) initiate a synchronization is $\Delta_i-|E_i|^{\kappa_i}=\frac{|E_i|^{\kappa_i } - 1}{|E_i| - 1}\approx|E_i|^{\kappa_i-1}$ ($|E_i|^{\kappa_i}$).
Therefore, the number of SI-states that cannot initiate a synchronization is bounded by $\prod^m_{i=1}|E_i|^{\kappa_i-1}$.
If we set $O_j$ to be the OS that initiates the synchronization, the number of CSI-states is at most $|E_j|^{\kappa_j}\prod_{i\in\mathcal{I}\setminus\{j\}}|E_i|^{\kappa_i-1}=|E_j|\cdot|E_j|^{\kappa_j-1}\cdot\prod_{i\in\mathcal{I}\setminus\{j\}}|E_i|^{\kappa_i-1}=|E_j|\cdot\prod^m_{i=1}|E_i|^{\kappa_i-1}$.
Thus, the total number of CSI-states can be bounded by $\sum_{j=1}^m(|E_j|\cdot\prod^m_{i=1}|E_i|^{\kappa_i-1})=\prod^m_{i=1}|E_i|^{\kappa_i-1}\cdot\sum_{j=1}^m|E_j|$.
Note that this is the worst case (e.g., when no observable events are shared among the OSs, which means $\sum_{j=1}^m|E_j|=|E_{\mathcal{I}}|$).
Since the number of CSI-states is bounded by $\Delta_c:=|E_{\mathcal{I}}|\cdot\prod^m_{i=1}|E_i|^{\kappa_i-1}$, the number of SI-states is bounded by $\Delta:=\Delta_c+\prod^m_{i=1}|E_i|^{\kappa_i-1}=(|E_{\mathcal{I}}|+1)\cdot\prod^m_{i=1}|E_i|^{\kappa_i-1}$.

Let $i_{u}\in\arg \max_{i\in\mathcal{I}}\kappa_i$ be one of the indices of OSs that can record the most events.
The length of the longest sequence of observable events that can be generated by the system and can initiate a synchronization is $l_u:=\kappa_{i_{u}}+\sum_{i\in\mathcal{I}\setminus\{i_{u}\}}(\kappa_i-1)$.

With the above notation at hand, the number of states in a CSS structure is at most  $l_u\cdot |X|^2+|X|+\Delta$ (where there are at most $l_u$ layers, each with at most $|X|^2$ states), the number of states with layer ``0'' is $|X|$, and that of SI-states is $\Delta$.
Hence, the state complexity of constructing a CSS structure is $O(l_u\cdot |X|^2+\Delta)$.
If there is only one initial state in a CSS structure (like an SS structure), each SI-state is reached by at most $|X|\cdot m$ transitions, and there are at most $m$ state estimation processes at each SI-state, resulting in complexity of $m\cdot|X|^2$.
Therefore, the computation complexity of constructing a CSS structure is $O((|X|\cdot m+m\cdot|X|^2)\cdot\Delta\cdot|X|)$, i.e., $O(\Delta\cdot m\cdot|X|^2(1+|X|))$.

Compared with the approach for DO-CSE in \cite{SunHadjicostisLi2023}, which has computational complexity of $O(|X|^2\cdot\prod_{i=1}^{m}(\kappa_i+1)\cdot m)$, the complexity at a synchronization step is reduced to $O(|X|^2)$, assuming that a CSS structure is precomputed.
There are at most $2^{|X|}$ states and $2^{|X|}\cdot\Delta_c$ transitions in a DO-observer. 
Therefore, the worst-case computational complexity of constructing a DO-observer is $O(2^{|X|}\cdot\Delta_c)$.
Similarly, the computational complexity of constructing the DO-based initial-state estimator is $O(2^{|X|^2}\cdot\Delta_c)$.

\section{Applications of DO-based State Estimators}

Opacity is an important notion that characterizes the security and privacy of information flow from a system to the public.
The secret of a system is  typically represented by a subset of states $X_S\subseteq X$.
In the setting considered in this paper, the system is observed by multiple observation sites that follow a DO-based protocol; thus, we consider the possibility that there exists a malicious outside observer (eavesdropper) that can track all partially ordered sequences of observations transmitted to the coordinator at each synchronization.
This consideration is natural if we treat the information flow from the OSs to the coordinator as insecure.
In this section, we define and analyze two kinds of opacity, called DO-based initial-state opacity (DO-ISO) and DO-based current-state-at-synchronization opacity (DO-CSSO).
In the end, we also provide some brief remarks on the verification of diagnosability under the DO-based protocol.

\subsection{Verification of DO-based Initial-State Opacity}
\begin{definition}
	(DO-ISO) A system $G=(X,E,\delta,X_0)$ under DO-based protocol $\Upsilon_{\mathbb{N}}$ and a set of secret initial states $X_S\subseteq X_0$ is said to be \textit{DO-based initial-state opaque} (w.r.t. $\Upsilon_{\mathbb{N}}$ and $X_S$)~if 
		\begin{align*}
		\forall s\in L(G,X_S), \exists s'\in L(G,X_0\backslash X_S):P_{\Upsilon_{\mathbb{N}}}(s)=P_{\Upsilon_{\mathbb{N}}}(s').
	\end{align*}
\end{definition}
Intuitively, DO-ISO\footnote{We use DO-ISO (DO-CSSO) as the acronym for both ``DO-based initial-state opaque/opacity'' (``DO-based current-state-at-synchronization opaque/opacity'').} requires that the intruder should never know with certainty if a system starts from a secret state after arbitrary number of synchronization steps.
With this understanding, the problem of the verification of DO-ISO can be reduced to the initial-state estimation problem, i.e., determining whether the initial-state estimates always contain a non-secret state.
Therefore, we have the following outcome.
\begin{theorem}\label{theorem-ini-1}
Given a system $G=(X,E,\delta,X_0)$ under DO-based protocol $\Upsilon_{\mathbb{N}}$ and  a set of secret initial states $X_S\subseteq X_0$, let $G^{\Upsilon_{\mathbb{N}}}_{i,obs}=(X^{i}_{obs},T_c,f^i_{obs},X^{i}_{0,obs})$ be its DO-based initial-state estimator. System $G$ is DO-ISO w.r.t. $\Upsilon_{\mathbb{N}}$ and $X_S$ if and only if for all $m\in X^{i}_{obs}$, $F(m)\nsubseteq X_S$ holds.
\end{theorem}
\begin{proof}
($\Rightarrow$) Suppose that system $G$ is DO-ISO. 
Then, for any $s\in L(G,X_S)$, there exists $s'\in L(G,X_0\backslash X_S)$ such that $P_{\Upsilon_{\mathbb{N}}}(s)=P_{\Upsilon_{\mathbb{N}}}(s')$.
Then, we know $\mathcal{E}^{i}(P_{\Upsilon_{\mathbb{N}}}(s),X_0\backslash X_S)\neq\emptyset$.
It is easy to verify that $\mathcal{E}^{i}(P_{\Upsilon_{\mathbb{N}}}(s),X_0)=\mathcal{E}^{i}(P_{\Upsilon_{\mathbb{N}}}(s),X_S)\cup\allowbreak \mathcal{E}^{i}(P_{\Upsilon_{\mathbb{N}}}(s),X_0\backslash X_S)$.
By Proposition~\ref{prop-ini}, $F(f^i_{0,obs}(X^i_{0,obs}, \allowbreak P_{\Upsilon_{\mathbb{N}}}(s)))\cap (X_0\backslash X_S)\neq\emptyset$. 
Then, $f^i_{0,obs}(X^i_{0,obs},P_{\Upsilon_{\mathbb{N}}}(s))\in X^{i}_{obs}$ and $F(f^i_{0,obs}(X^i_{0,obs},P_{\Upsilon_{\mathbb{N}}}(s)))\nsubseteq X_S$ hold, which completes the proof in this direction.
    
(``$\Leftarrow$'') Suppose that for any $m\in X^{i}_{obs}$, $F(m)\nsubseteq X_S$. 
According to Proposition~\ref{prop-ini}, there exists $s\in L(G,X_0)$ such that $m=f^i_{0,obs}(X^i_{0,obs},P_{\Upsilon_{\mathbb{N}}}(s))$ and $\mathcal{E}^{i}(P_{\Upsilon_{\mathbb{N}}}(s),X_0)=F(m)$.
Then, it gives $\mathcal{E}^{i}(P_{\Upsilon_{\mathbb{N}}}(s),X_0)\nsubseteq X_S$, indicating that there exists $s'\in L(G,X_0\setminus X_S)$ such that  $P_{\Upsilon_{\mathbb{N}}}(s)=P_{\Upsilon_{\mathbb{N}}}(s')$, i.e., $G$ is DO-ISO.
\end{proof}

In \cite{WuLafortune2013}, the notion of reversed automaton is introduced to verify centralized initial-state opacity. 
However, this method cannot be implemented directly to verify DO-ISO here, as the coordinator receives a set of partially ordered sequences at a time, rather than a single observation.
Nevertheless, inspired by this method, we present an exponential time approach to verify DO-ISO, compared with the one in Theorem \ref{theorem-ini-1}, which has doubly-exponential complexity. 
Instead of introducing the reversed automaton of the original system, we directly define the notion of a synchronization-reversed observer as follows.
\begin{definition}\label{def-sro}
	The synchronization-reversed observer of a system $G$ is a deterministic finite-state automaton $G^{\Upsilon_{\mathbb{N}}}_{R,obs}=Ac(X_{R,obs},T_c, f_{R,obs}, X)$, where
	\begin{itemize}
		\item $X_{R,obs}\subseteq 2^X\setminus\{\emptyset\}$ is the set of states;
		\item $T_c$ is the set of CSI-states of system $G$'s feasible CSS structure;
		\item $X\in X_{R,obs}$ is the initial state;
		\item $f_{R,obs}: X_{R,obs}\times T_c\rightarrow X_{R,obs}$ is the transition function defined as: for any $q\in X_{R,obs}$, $\tau\in T_c$,
		\begin{align*}
			f_{R,obs}(q,\tau)=\{x\in X|\exists x'\in q,(x,x')\in M(\tau)\}.
		\end{align*}
	\end{itemize}
\end{definition}

Intuitively, the differences between DO-observer and synchronization-reversed observer lie in the transition function and the initial state.
Function $f_{R,obs}$ computes the next state by reversing the state pairs in $M(\tau)$, i.e., it takes the union of the first components in $M(\tau)$ whose corresponding second components belong to $q$.
The initial state $X$ is the entire state space of system $G$.
Like function $f_{obs}$, the domain of $f_{R,obs}$ can be extended to $X_{R,obs}\times T^*_c$ in the same manner.
Given a sequence of CSI-states $\iota\in T^*_c$, $\iota_R$ is the reversed sequence of $\iota$.
We arrive at the following conclusion based on $G^{\Upsilon_{\mathbb{N}}}_{R,obs}$.

\begin{lemma}\label{lemma-reverse}
Let $G^{\Upsilon_{\mathbb{N}}}_{R,obs}$ be the synchronization-reversed observer of a system $G$ and for any $s\in L(G)$, let $\iota=P_{\Upsilon_{\mathbb{N}}}(s)$ be the sequence of CSI-states received by the coordinator. 
Then, it holds $\mathcal{E}^{i}(\iota,X_0)=f_{R,obs}(X,\iota_R)\cap X_0$.
\end{lemma}
\begin{proof}
	By using the extended definition of $f_{R,obs}$, for any sequence of CSI-states $\iota=\tau_n\tau_{n-1}\dots\tau_1\in L(G^{\Upsilon_{\mathbb{N}}}_{obs})$, we find $f_{R,obs}(X,\iota_R)=\{x_n\in X|\exists x_0, x_1, x_2, \dots, x_n\in X,(x_n,x_{n-1})\in M(\tau_n),(x_{n-1},x_{n-2})\in M(\tau_{n-1}),\dots,\\(x_{1},x_{0})\in M(\tau_{1})\}$. 
	Then, $f_{R,obs}(X,\iota_R)\cap X_0=\{x_n\in X_0|\exists x_0, x_1, x_2, \dots, x_{n-1}\in X,x_n\in X_0,(x_n,x_{n-1})\in M(\tau_n),(x_{n-1},x_{n-2})\in M(\tau_{n-1}),\dots,(x_{1},x_{0})\in M(\tau_{1})\}$. 
	Combining Proposition \ref{pro-current-string}.1) with the definition of $\mathcal{E}^{i}(\iota,X_0)$, we have $f_{R,obs}(X,\iota_R)\cap X_0=\mathcal{E}^{i}(\iota,X_0)$, which completes the proof.
\end{proof}

\begin{theorem}
	Given a system $G=(X,E,\delta,X_0)$ under the DO-based protocol $\Upsilon_{\mathbb{N}}$ and  a set of secret states $X_S\subseteq X_0$, let $G^{\Upsilon_{\mathbb{N}}}_{R,obs}=Ac(X_{R,obs},T_c, f_{R,obs}, X)$ be its synchronization-reversed observer. System $G$ is DO-ISO w.r.t. $\Upsilon_{\mathbb{N}}$ and $X_S$ if and only if 
\begin{align*}
\forall q\in X_{R,obs}:q\cap X_0\neq\emptyset\Rightarrow (q\cap X_0)\nsubseteq X_S.
\end{align*}
\end{theorem}
\begin{proof}
According to Theorem \ref{theorem-ini-1}, the system is DO-ISO if and only if its initial state estimates should never be a subset of $X_S$. 
We proceed by contrapositive.

($\Rightarrow$)
Assume that there exists $q\in X_{R,obs}$ such that $q\cap X_0\allowbreak\neq\emptyset$ and $ (q\cap X_0)\subseteq X_S$.
By the definition of $G^{\Upsilon_{\mathbb{N}}}_{R,obs}$ and Lemma \ref{lemma-reverse}, there exists $s\in L(G)$ such that $\mathcal{E}^{i}(\iota,X_0)=q\cap X_0$ where $\iota=P_{\Upsilon_{\mathbb{N}}}(s)$ and $q=f_{R,obs}(X,\iota_R)$.
Therefore, $\mathcal{E}^{i}(\iota,X_0)\subseteq X_S$, indicating that the system is not DO-ISO.

($\Leftarrow$)
If the system is not DO-ISO, there exists $s\in L(G)$ such that $\mathcal{E}^{i}(\iota,X_0)\subseteq X_S$ where $\iota=P_{\Upsilon_{\mathbb{N}}}(s)$.
Due to Lemma \ref{lemma-reverse}, we have $\mathcal{E}^{i}(\iota,X_0)=f_{R,obs}(X,\iota_R)\cap X_0$, suggesting that there exists $q=f_{R,obs}(X,\iota_R)\in X_{R,obs}$ such that $q\cap X_0\neq\emptyset$ and $ (q\cap X_0)\subseteq X_S$.
This completes the proof.

\end{proof}

\begin{remark}
	Both approaches analyzed above are based on the structure of the system's CSS structure.
	Assuming that the CSS structure is computed beforehand, the computational complexity of constructing a synchronization-reversed observer is equivalent to that of computing a DO-observer, i.e., $O(\Delta_c\cdot 2^{|X|})$.
	The approach utilizing the DO-based initial-state estimator is of high complexity, as it stores more information, namely, the sets of pairs of possible initial and current states. 
	When verifying DO-ISE, one can choose the appropriate method based on the underlying requirements.
\end{remark}

\begin{figure}[htbp]
	\centering
	\includegraphics[scale=0.7]{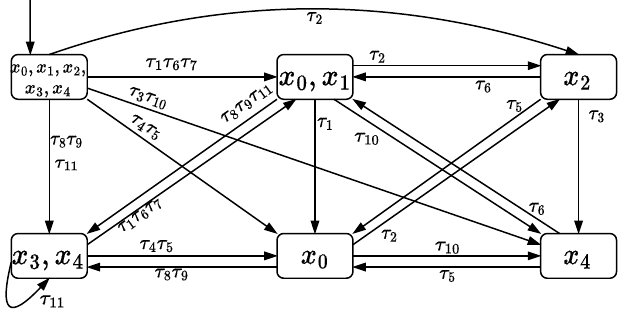}
	\caption{The synchronization-reversed observer $G^{\Upsilon_{\mathbb{N}}}_{R,obs}$ for $G$. The transitions labeled as ``$\tau$'' correspond to the labels shown in Fig.~\ref{fig-CSS-structure-1}.}
	\label{fig-sro}
\end{figure}

\begin{example}
	Consider again the system $G$ in Fig.~\ref{fig-example-1}, where $X_0=X$ and $X_S=\{x_0\}$.
	The feasible CSS structure of $G$ is shown in Fig.~\ref{fig-CSS-structure-1}.
	Based on Definition \ref{def-sro}, the synchronization-reversed observer is constructed and shown in Fig.~\ref{fig-sro}.
    For example, since $(x_4,x_2)\in M(\tau_3)$, there exists a transition relation $(\{x_2\},\tau_3,\{x_4\})\in f_{R,obs}$ in $G^{\Upsilon_{\mathbb{N}}}_{R,obs}$.
    By taking $\iota_R=\tau_6\tau_1$, after observing the sequence of CSI-states $\iota=\tau_1\tau_6$, the outside observer will know for sure that the system initial state was $x_0$.
    Therefore, this system is not DO-ISO w.r.t. $\Upsilon_{\mathbb{N}}$ and $X_S$.\hfill\rule{1ex}{1ex}
\end{example}

\subsection{Verification of DO-based Current-State-at-Synchronization Opacity}
The notion of centralized current-state opacity requires that the system should possibly reside in a non-secret state at any given instant.
However, this property cannot be directly applied to a system under the DO-based protocol, where several events, unknown to the coordinator (and the outside observer) until the next synchronization, will be executed in the system after each synchronization.
To address this, a new privacy requirement is introduced, focusing solely on the timing immediately after each synchronization.
Since the synchronization instant represents the most clearly defined and predictable moment in the system's state, it is crucial to ensure that the system's state information at these instants remains confidential.
As noted in Remark~\ref{remark-aftersync}, addressing the privacy/security considerations between synchronizations is left for future research due to the complexity of this topic, which cannot be fully covered here within the allowed space.
	\begin{definition}
		(DO-CSSO) A system $G=(X,E,\delta,X_0)$ under DO-based protocol $\Upsilon_{\mathbb{N}}$ and a set of secret states $X_S\subseteq X$ is said to be \textit{DO-based current-state-at-synchronization opaque}  (w.r.t. $\Upsilon_{\mathbb{N}}$ and $X_S$) if the intruder can never know with certainty if the system resides in a secret state immediately following a synchronization, i.e.,
			\begin{multline*}
				\forall x_0\in X_0, \forall s\in L(G,x_0): \operatorname{UR}(\delta(x_0,\tilde{s}))\cap X_S\neq\emptyset
				\Rightarrow\\
				\exists x_0'\in X_0, \exists s'\in L(G,x_0'): P_{\Upsilon_{\mathbb{N}}}(s)=P_{\Upsilon_{\mathbb{N}}}(s')\land\\
				\operatorname{UR}(\delta(x_0',\tilde{s'}))\cap (X\backslash X_S)\neq\emptyset.
			\end{multline*}
	\end{definition}
	In DO-CSSO, for any string $s\in L(G,x_0)$, only the string $\tilde{s}$ is used to describe the privacy requirements.
	If the system may reside in a secret state (i.e., $\operatorname{UR}(\delta(x_0,\tilde{s}))\cap X_S\neq\emptyset$), there should be another string $s'\in L(G,x_0')$ that has the same DO-projection as $s$ (i.e., $P_{\Upsilon_{\mathbb{N}}}(s)=P_{\Upsilon_{\mathbb{N}}}(s')$) and leads the system to a non-secret state (i.e., $\operatorname{UR}(\delta(x_0',\tilde{s'}))\cap (X\backslash X_S)\neq\emptyset$).
	In this situation, the DO-observer can be used to verify this property by simply checking if there are non-secret states within each state of the DO-observer.
\begin{theorem}\label{theorem-do-CSSO}
	Given a system $G=(X,E,\delta,X_0)$ under DO-based protocol $\Upsilon_{\mathbb{N}}$ and  a set of secret states $X_S\subseteq X$, let $G^{\Upsilon_{\mathbb{N}}}_{obs}=Ac(X_{obs},T_c,f_{obs},X_{obs,0})$ be its DO-observer. System $G$ is DO-CSSO w.r.t. $\Upsilon_{\mathbb{N}}$ and $X_S$ if and only if for any $q\in X_{obs}$, $q\nsubseteq X_S$.
\end{theorem}
\begin{proof}
According to the definition of DO-CSE, it is trivial to deduce that system $G$ is DO-CSSO if for all $s\in L(G)$, $\mathcal{E}^c(P_{\Upsilon_{\mathbb{N}}}(s),X_0)\cap (X\backslash X_S)\neq\emptyset$.
We proceed with the contrapositive.

($\Rightarrow$) Assume that there exists $q\in X_{obs}$, $q\subseteq X_S$.
According to the definition of DO-observer and Proposition~\ref{pro-obs}, there exists $s\in L(G)$, $q=\mathcal{E}^c(P_{\Upsilon_{\mathbb{N}}}(s),X_0)$.
Thus, it holds $\mathcal{E}^c(P_{\Upsilon_{\mathbb{N}}}(s),X_0)\cap (X\backslash X_S)=\emptyset$.
The system is not DO-CSSO.

($\Leftarrow$) If the system is not DO-CSSO, there exists $s\in L(G)$, $\mathcal{E}^c(P_{\Upsilon_{\mathbb{N}}}(s),X_0)\cap (X\backslash X_S)=\emptyset$.
According to the definition of DO-observer and Proposition~\ref{pro-obs}, there exists $q\in X_{obs}$, $q\subseteq X_S$.
This completes the proof.
\end{proof}
\begin{example}
Let us consider the system in Example~\ref{example-2} where the set of initial states is $\{x_0,x_1\}$ and the set of secret states is $X_S=\{x_2\}$.
Its DO-observer is shown in Fig.~\ref{DO-observer-fig}.
According to Theorem~\ref{theorem-do-CSSO}, system $G$ is not DO-CSSO w.r.t. $\Upsilon_{\mathbb{N}}$ and $X_S=\{x_2\}$.
If the sequence of CSI-states $\tau_5\tau_3$ is sent to the coordinator, the outside observer can determine that the system is in secret state $x_2$.
\end{example}

\subsection{Discussion of Diagnosability Under the DO-based Protocol}
It is shown in \cite{Hadjicostis2020} that the standard fault diagnosis problem of a DES can be reduced to a state isolation/estimation problem.
By properly refining the state space of the original system to incorporate the occurrence of fault events, similar to the procedure in \cite{Sampath1995}, the refined system is constructed.
The states in the refined system  are classified into two categories: those augmented with the symbol ``$F$'', indicating that the system has encountered one or more fault events before reaching these states, and those not augmented with ``$F$'', implying the system has not executed any fault events up to that point.
Consequently, this problem is reduced to finding so-called $F$-indeterminate cycles \cite{Sampath1995} in the observer of the refined system.
Furthermore, the results in \cite{YooLafortune2002, JiangHuang2001} show that the verification of diagnosability can be tested in polynomial time by constructing a parallel composition structure called a verifier.

Diagnosability under the DO-based protocol was preliminarily studied in \cite{Hadjicostis2020} through a structure called $GSDS_m$, built using a partial-order-based approach.
However, this method incurs significant computational cost.
In contrast, the CSS structure introduced in this paper enables the use of standard diagnosis techniques in the DO-based setting more efficiently.
Unlike other decentralized architectures where the coordinator receives local state estimates or decisions, the DO-based coordinator receives partially ordered system behaviors (CSI-states).
As shown in Proposition~\ref{pro-current-string}, the CSS structure captures the evolution of system states under all possible CSI-states, each corresponding to specific system behaviors.
This enables the coordinator to derive state estimates from CSI-states, thereby allowing the construction of a DO-observer.
Although some observable events may not be immediately available to the coordinator, state estimates are updated upon their reception at the next synchronization.
In this regard, the DO-based protocol exhibits both centralized and decentralized characteristics, and the verification of diagnosability can follow an approach analogous to the centralized setting.
To stay focused, we omit the formal construction details in this paper.
\section{Conclusions}
In this paper, we present a relatively efficient approach to the problem of state estimation for a DES under a decentralized observation architecture where observation sites send to the coordinator the sequences of observations they recorded.
We also discuss how to verify state-isolation DO-based properties using this approach.
To this end, we first provide the formal description of the DO-based protocol.
Then, we interpret a given DO-based strategy into a structure called \textit{Complete Synchronizing Sequence structure} (CSS structure), which includes the relations between states under any SI-state.
Based on the CSS structure, the DO-observer and DO-based initial state estimator are defined and constructed, with the processes simply involving the union of states without exploring possible system behavior at each synchronization step.
It is shown that, using the CSS structure, DO-based initial-state opacity and DO-based current-state-at-synchronization opacity can be verified similarly to the standard centralized observation architecture.

In the future, we plan to utilize the CSS structure to develop methods for detecting fault events in the presence of local and global errors, i.e., when the synchronization information received by the coordinator is tampered.
It is also worth exploring the privacy/security considerations between synchronizations.

\end{document}